\documentclass[a4paper,11pt]{article}
\usepackage[margin=1in]{geometry}
\usepackage{complexity} 
\usepackage{enumerate,paralist}
\usepackage{amsmath}
\usepackage{amssymb,amsthm,complexity}
\usepackage{authblk}
\usepackage{todonotes}

\newcommand{\sv}[1]{}

\usepackage{vmargin}
\setmarginsrb{1in}{1in}{1in}{1in}{0mm}{0mm}{0mm}{7mm}

\usepackage{mathtools}
\usepackage{thmtools}
\usepackage{thm-restate}
\usepackage{hyperref}
\usepackage{cleveref}
\usepackage{complexity}
\usepackage{authblk}
\theoremstyle{plain}
\usepackage{todonotes}
\usepackage{microtype}
\usepackage{multirow}
\usepackage{amsmath,amssymb}
\usepackage{comment}
\usepackage{enumerate}
\usepackage{xspace}
\usepackage{listings}
\usepackage{color}
\usepackage{cite}
\usepackage{graphicx}
\usepackage{subcaption}
\RequirePackage{fancyhdr}
 \usepackage{xcolor}
\usepackage{boxedminipage}
\usepackage{algorithm2e}
\usepackage{tabularx}
\usepackage{booktabs}

\newcommand{\defparproblem}[4]{
  \vspace{3mm}
\noindent\fbox{
  \begin{minipage}{.95\textwidth}
  \begin{tabular*}{\textwidth}{@{\extracolsep{\fill}}lr} \textsc{#1}\\ \end{tabular*}
  {\bf{Input:}} #2  \\
  {\bf{Parameter:}} #3 \\
  {\bf{Question:}} #4
  \end{minipage}
  }
  \vspace{2mm}
}

\newcommand{\defproblem}[3]{
  \vspace{3mm}
\noindent\fbox{
  \begin{minipage}{.95\textwidth}
  \begin{tabular*}{\textwidth}{@{\extracolsep{\fill}}lr} #1  \\ \end{tabular*}
  {\bf{Input:}} #2  \\
  {\bf{Question:}} #3
  \end{minipage}
  }
  \vspace{2mm}
  }

\DeclarePairedDelimiter\ev{\langle}{\rangle}
\DeclarePairedDelimiter\ceil{\lceil}{\rceil}

\newcommand{\AAA}{{\mathcal A}}
\newcommand{\BB}{{\mathcal B}}

\newcommand{\cD}{\mathcal{D}}

\newcommand{\cO}{{\mathcal O}}
\newcommand{\cI}{{\mathcal I}}
\newcommand{\cX}{{\mathcal X}}
\newcommand{\cF}{{\mathcal F}}

\newcommand{\bbF}{{\mathbb F}}

\newcommand{\RBDS}{{\sc Red Blue Dominating Set}}
\newcommand{\UAQ}{{\sc User Authorization Query}}

\newcommand{\mulbiclique}{\sc Multicolored Biclique}

\newtheorem{theorem}{\bf Theorem}
\newtheorem{definition}{\bf Definition}
\newtheorem{proposition}{\bf Proposition}

\newtheorem{reduction rule}{\bf Reduction Rule}
\newtheorem{branching rule}{\bf Branching Rule}
\newtheorem{lemma}{\bf Lemma}
\newtheorem{construction}{Construction}

\newtheorem{innercustomthm}{Procedure}
\newenvironment{reduction procedure}[1]
{\renewcommand\theinnercustomthm{#1}\innercustomthm}
  {\endinnercustomthm}

\newtheorem{newinnercustomthm}{Reduction Rule}
\newenvironment{reduction}[1]
{\renewcommand\thenewinnercustomthm{#1}\newinnercustomthm}
  {\endnewinnercustomthm}

\title{Towards Better Understanding of User Authorization Query Problem via Multi-variable Complexity Analysis}

\date{}
\author{Jason Crampton, Gregory Gutin, and Diptapriyo Majumdar}
\affil{Royal Holloway, University of London, Egham, United Kingdom}
\begin{document}

\maketitle

\begin{abstract}
User authorization queries in the context of role-based access control have attracted considerable interest in the last 15 years.
Such queries are used to determine whether it is possible to allocate a set of roles to a user that enables the user to complete a task, in the sense that all the permissions required to complete the task are assigned to the roles in that set.
Answering such a query, in general, must take into account a number of factors, including, but not limited to, the roles to which the user is assigned and constraints on the sets of roles that can be activated.
Answering such a query is known to be NP-hard.
The presence of multiple parameters and the need to find efficient and exact solutions to the problem suggest that a multi-variate approach will enable us to better understand the complexity of the user authorization query problem (UAQ).

In this paper, we establish a number of complexity results for UAQ.
Specifically, we show the problem remains hard even when quite restrictive conditions are imposed on the structure of the problem.
Our FPT results show that we have to use either a parameter with potentially quite large values or quite a restricted version of UAQ. Moreover, our second FPT algorithm is complex and requires sophisticated, state-of-the-art techniques.
In short, our results show that it is unlikely that all variants of UAQ that arise in practice can be solved reasonably quickly in general.
\end{abstract}




\maketitle

\section{Introduction}
\label{sec:uaq}

In Role-Based Access Control (RBAC), permissions are not assigned to users directly.
A user is assigned roles and roles are assigned permissions.
Thus, a user $u$ is authorized for those permissions that are assigned to at least one role to which $u$ is assigned.

A user interacts with an RBAC system by activating some subset of the roles to which she is assigned.
In certain situations, it is useful to be able to identify the particular subset of roles a user needs to activate in order to complete a task that requires specific permissions.
A user authorization query seeks to find a set of roles that is suitable for a given set of permissions.
A substantial body of work has established that user authorization queries are, in general, hard to solve~\cite{DuJo06,ChCr09}.

Generally speaking, it is important from the end-user perspective to answer access control queries as quickly as possible.
Moreover, it is important that the answers to those queries are correct, otherwise a user may not be able to complete a task (because too few permissions are assigned to the solution's role set) or security breaches may occur (in the case of too many permissions).
Thus, it is desirable to find algorithms to solve user authorization queries that are exact and as fast as possible.
Existing {\sf NP}-hardness results  \cite{ChCr09,DuJo06} show that such algorithms are unlikely to exist for all user authorization queries.
However, such queries have several  {numerical} parameters, some of which may be small in all instances of practical interest.
Hence, it is worth exploring such queries from the perspective of fixed-parameter tractability (a short introduction to parameterized algorithms and complexity is given in Section \ref{sec:prel}). 

As for other practical problems, SAT and other solvers for  \UAQ\ (UAQ) have been tested and compared with other solvers using benchmarks.
Recently, Armando et al. \cite{ArmandoGT20} suggested a methodology to evaluate existing benchmarks for UAQ and to guide
the design of new ones. The methodology is based on the use of fixed-parameter tractable ({\sf FPT}) algorithms, 
{i.e.,  algorithms which run in time $f(k)N^{c},$ where $N$ is the size of the problem instance, $k$ is a parameter or sum of several parameters, $f$ is a function depending only on $k$, and $c$ is a constant.
Note that when $k$ is fixed, then the algorithm's running time becomes polynomial of constant degree. Often $c$ is small, but for an {\sf FPT} algorithm to be practical, $k$ should be relatively small and $f$ relatively slow growing.} 
Armando et al. observed that ``the proposed methodology 
will yield different, improved results as soon as new complexity results will become available. Thus the benchmarks proposed and used in this paper could be improved consequently.'' 
In this paper, we prove some parameterized tractability results for restricted versions of UAQ that can be applied in that methodology. 

On the other hand, we show that queries for some quite constrained RBAC configurations remain hard, even from the parameterized complexity point of view, {i.e., it is highly unlikely that there are {\sf FPT} algorithms for them.}
Our parameterized intractability results indicate that to obtain 
{\sf FPT} algorithms for UAQ with small parameters we have to consider quite restricted versions of UAQ.

Two well-known access control problems -- UAQ and the Workflow Satisfiability Problem (WSP) -- are ``equally" intractable from the classical complexity point of view (i.e., {\sf NP}-complete). However, from a more fine-grained parameterized complexity point of view, the time complexities of the two problems, even in their most basic forms, are quite different. 
While Wang and Li \cite{WaLi10} and subsequent research (e.g.,  \cite{CoCrGaGuJo14,CrGaGuJoWa16,GutinK20}) showed that basic and more advanced versions of WSP are {\sf FPT}, our hardness results in this paper clearly demonstrate that several parameterized basic versions of UAQ are still intractable. Moreover, we show in Section \ref{sec:w-hardness-k-ij} that even a quite restrictive version of UAQ can require advanced algorithmic methods 
to design an {\sf FPT} algorithm, which was not the case with WSP. 
Thus, it is likely that the search for practical algorithms to solve UAQ will be less successful than it was for WSP (see, e.g.,~\cite{BertolissiSR18,CohenCGGJ16,KarapetyanPGG19}).

\subsection{User Authorization Query Problem and Its Reduction}
 

An RBAC \emph{configuration} has the form $\rho = ( (R,\leq), P, RP )$, where $(R,\leq)$ is a \emph{role hierarchy} modeled as a partially ordered set of roles, and $RP \subseteq R \times P$ is an assignment of roles to permissions.
A role $r \in R$ is \emph{authorized} for the set of permissions $P(r) = \{p \in P \mid (p,r') \in RP, r' \leq r\}.$
Then a user $u$, authorized for a set of roles $R' \subseteq R$, is authorized for the set of permissions $P(R') = \bigcup_{r \in R'} P(r)$.
The {\em role-permission} graph (RPG) is the bipartite graph $(R \cup P, RP)$ with partite sets $R$ and $P$ and edge set $RP$. (It is assumed that $R$ and $P$ are disjoint.)
A (dynamic) \emph{separation of duty constraint} has the form $\langle X, t \rangle$, where $1 \leq t \leq |X|$ for $X \subseteq R$.
The semantics of an SoD constraint is that no user can activate $t$ or more roles from $X$ (thus restricting the sets of roles that may provide a solution to the user authorization query). 
We can now formulate a decision version of the problem studied in this paper. 

\defproblem{{\UAQ}}%
            {An RBAC policy $((R,\leq),P,RP)$, $P_{lb}, P_{ub} \subseteq P$, a set of constraints $\cD$, and integers $k_r$, $k_p.$}
            {Does there exist a {\em solution}, which is set of roles $R_s \subseteq R$ such that $|R_s| \leqslant k_r$, $R_s$ satisfies all constraints in $\cD$, $P_{lb} \subseteq P(R_s) \subseteq P_{ub}$ and $|P(R_s) \setminus P_{lb}| \leqslant k_p$? If the answer is yes, find a solution.}
            
Note that there are search versions of the problem \cite{WiQaLi09,ArmandoGT20}. 
Wickramaarachchi et al.~\cite{WiQaLi09}, for example, do not specify $k_p$ and $k_r$.
Instead, a problem instance includes an objective, which takes one of two values -- $\it min$ or $\it max$ -- and specifies the required nature of a solution set.
Specifically, given $((R,\leq),P,RP)$, $P_{lb}$, $P_{ub}$, $\cD$ and $\it min$ (respectively, $\it max$), find $R' \subseteq R$ such that 
\begin{enumerate}
 \item $P_{lb} \subseteq P(r) \subseteq P_{ub}$, 
 \item $R'$ satisfies all constraints in $\cD$, and
 \item for any $R''$ that satisfies conditions 1 and 2, we have $|R'| \leq |R''|$ (respectively, $|R'| \geq |R''|$).
\end{enumerate}

            
We say that {\UAQ} (UAQ) is {\em non-hierarchical} if no partial order is defined on the set of roles. 
Moffet~\cite{Moffet98} has argued that role hierarchies are not necessarily appropriate structures for access control (see also \cite{MoffetL99}).
Moreover, it is always possible to eliminate a role hierarchy by assigning the (inherited) permissions of junior roles explicitly to more senior roles.
Hereafter, we will consider only non-hierarchical UAQs.

We will use the following reduction from the UAQ instance to an equivalent instance, i.e., both instances are either yes-instances or no-instances. 

\vspace{2mm}

\begin{reduction}{0}
\label{rule:initial-assumption}
(i) delete any role $r$ such that $P(r) \cap P_{lb} = \emptyset $ (since $r$ cannot contribute to a set of roles that satisfies the instance);\\
(ii) delete any role $r$ such that $P(r) \setminus P_{ub} \ne \emptyset$ (since $r$ contributes permissions that are not allowed);\\
(iii) for any deleted role $r$  and any  SoD constraint $\langle X,t\rangle$ with $r\in X$, replace the constraint by $\langle X\setminus \{r\},t \rangle$;\\
(iv) for any deleted role $r$, remove all pairs of the form $(r,p)$ from $RP$;\\
(v) delete any constraint $\ev{X, t}$ in which $|X| < t$.
\end{reduction}
	
 	


One consequence of applying the above reduction rule is that we may assume without loss of generality that $P_{ub} = P$. 	
Thus, in this paper a UAQ instance will be written as a tuple $(R,P,RP,P_{lb},\cD,k_r,k_p)$.

\subsection{Our results}

In Section \ref{sec:hard}, we present parameterized intractability results for UAQ when there are no SoD constraints. 
{These results are based on a hierarchy of parameterized intractability classes, which we describe in more detail in Section 2. Informally, parameterized problems which admit {\sf FPT} algorithms form the tractable class {\sf FPT}. There is an infinite number of parameterized intractability classes {\sf W}[i], $i \ge 1$, such that {\sf FPT}$\subseteq${\sf W}[1]$\subseteq${\sf W}[2]$\subseteq \dots$. A problem belongs to {\sf W}[i] if it can be reduced to one of the hardest problems in {\sf W}[i]. It is widely believed that {\sf FPT}$\ne${\sf W}[1] (and hence {\sf FPT}$\ne${\sf W}[i] for any $i$). In particular, a parameterized problem proved to be {\sf W}[i]-hard is highly unlikely to admit an {\sf FPT} algorithm.}

We prove that if $P_{\ell b}  = P$ then UAQ parameterized by $k_r$ is {\sf W}[2]-hard and if 
$k_r = |R|$ (i.e., there is no restriction on the size of UAQ solution) then UAQ parameterized by $k_p$ is {\sf W}[2]-hard.
Note that these results strengthen the classical {\sf NP}-hardness results of Du and Joshi \cite{DuJo06} (for $\cD=\emptyset$ and $P_{\ell b}  = P$) and of Chen and Crampton \cite{ChCr09} (for $\cD=\emptyset$ and $k_r = |R|$).  
In the same section, we also prove that UAQ parameterized by $|P_{\ell b}| + k_p$ is {\sf W}[1]-hard even under the following restrictions:
(i)~every role is authorized for at most three permissions from $P$ and at most two permissions from $P \setminus P_{\ell b}$, and
(ii)~every two roles are authorized for at most two common permissions. 
This result shows that we have parameterized intractability even for UAQ without SoD constraints, with quite a restrictive structure imposed on RPG, and parameterized by a parameter which can be quite large due to $|P_{\ell b}|$ (see, e.g., Table 1 in~\cite{ArmandoGT20}).

Let $R_2$ be the set of roles that are authorized for at least one permission outside $P_{\ell b}$ and let $\hat{r}=|R_2|.$
In Section~\ref{sec:hard} we also prove that UAQ with $\cD=\emptyset$ and $k_r = |R|$ admits an algorithm of running time $\cO^*(2^{\hat{r}})$.
Note that this algorithm is asymptotically faster than an algorithm of runtime $\cO^*(2^{|R|})$ introduced in~\cite{ZhangJ08} and studied in~\cite{WiQaLi09,MoTr12}, so it can be used for producing new UAQ benchmarks with $\cD=\emptyset$ and $k_r = |R|$.

Let ${\hat k}=|P \setminus P_{\ell b}|.$
In Section \ref{sec:w-hardness-k-ij}, we study UAQ such that $|P(r) \cap P(r')| \leqslant 1$ for all $r\neq r' \in R$. We prove that this form of UAQ, parameterized by $k_r + {\hat k}$, is {\sf W}[1]-hard even when every SoD constraint $\ev{X,t}$ in $\cD$ has $|X| = t = 2$. Note that in this result the parameter can be significantly larger than in the one above for $P_{\ell b}  = P$ but we do allow SoD constraints. 
This result can be easily extended to the case when for fixed integers $\alpha\ge 2$ and $\beta\ge 2$, 
no $\alpha$ roles are collectively authorized for $\beta$ permissions.

In Section \ref{sec:k-i-j-free}, 
we study a more restrictive UAQ problem satisfying the following conditions for fixed integers $\alpha\ge 2$ and $\beta\ge 2$:
(i)~for every set of $\alpha$ roles, $\{r_1,\ldots,r_{\alpha}\}$, $|P(r_1) \cap P(r_2) \cap \cdots \cap P(r_{\alpha})| < \beta$;
(ii)~there is a constant $c$ such that for every SoD constraint $\ev{X, t}$ we have $|X|\le c$; and
(iii)~for every pair $\ev{X_1, t_1}, \ev{X_2, t_2}$ of SoD constraints, we have  $X_1 \cap X_2 = \emptyset$.
We prove that this problem parameterized by $k_r + \hat k$ is {\sf FPT} and admits an algorithm of running time $\cO^*(2^{\cO(k_r^{\alpha} + \hat k)})$.
The design of this algorithm incorporates several algorithmic tools: reduction rules, branching rules and advanced dynamic programming which uses representative families on matroids. 
Hence, we believe that it is unlikely that (simple) {\sf FPT} algorithms exist when the parameter of interest is small, even for significantly restricted versions of UAQ, in sharp contrast to WSP, for which efficient {\sf FPT} algorithms exist for most instances of practical interest~\cite{CoCrGaGuJo14,KarapetyanPGG19}.

For ease of reference, we summarize the notation we use in the paper and our results in Tables~\ref{t1} and~\ref{t2}, respectively.
 
\begin{table}[h]
\centering
  \caption{Summary of notation used in the paper}\label{t1}
 \begin{tabular}{ll}
 \toprule
  \bf Notation & \bf Meaning \\
  \midrule
  $R$ & Set of roles \\
  $P$ & Set of permissions \\
  $\it RP \subseteq R \times P$ & Role-permission assignment relation \\
  $P(r) \subseteq P$ & Set of permissions assigned to role $r$ \\
  $\langle X,t \rangle$, $X \subseteq R$, $t \in \mathbb{N}$ & Separation-of-duty constraint \\
  $\mathcal{D}$ & Set of separation-of-duty constraints \\
  \midrule
  $P_{\ell b} \subseteq P$ & Required set of permissions in UAQ solution \\
  $k_r \in \mathbb{N}$ & Maximum number of roles in UAQ solution \\
  $k_p \in \mathbb{N}$ & Maximum number of permissions outside $P_{\ell b}$ in UAQ solution \\
  $\hat{r} \in \mathbb{N}$ & Number of roles assigned to at least one permission not in $P_{\ell b}$ \\
  $\hat{k} \in \mathbb{N}$ & Number of permissions not in $P_{\ell b}$ \\
  \midrule
  $G = (V,E)$ & Graph $G$ with vertex set $V$ and edge set $E$ \\
  $N(v)$, $v \in V$ & Set of neighbors of $v$ in $G = \langle V,E \rangle$ \\
  $N(S)$, $S \subseteq V$ & Set of neighbors of vertices in $S$ \\
  $A \uplus B$ & Union of disjoint sets $A$ and $B$ \\
  $G =(A \uplus B,E)$ & Bipartite graph $G$, $xy \in E$ iff $x \in A$ and $y \in B$ \\
  $K_{\alpha,\beta} = (A \uplus B,E)$ & Complete bipartite graph, $|A| = \alpha$, $|B| = \beta$ \\
  \midrule
  $[k]$ & Set of integers $\{1,\dots,k\}$ \\
  $\mathcal{O}^*(f(k))$ & $\mathcal{O}(f(k)p(k))$, where $p$ is some polynomial \\
  \bottomrule
 \end{tabular}
\end{table}

\vspace*{\baselineskip}

\begin{table}[h]\centering 
 \caption{Summary of our results}\label{t2}
 \begin{tabular}{cccccl}
  \toprule
  \multicolumn{4}{c}{\bf Input restrictions} & \multirow{2}{*}{\bf Parameter} & \multirow{2}{*}{\bf Complexity}  \\
  \cmidrule{1-4}
  $k_r$ & $k_p$ & RPG & $\mathcal{D}$ & & \\
  \midrule
  $|R|$ & -- & -- & $\emptyset$ & $k_p$ & W[2]-hard \\
  \cmidrule{1-4}
  -- & $0$ & -- & $\emptyset$ & $k_r$ & W[2]-hard \\
  \cmidrule{1-4}
  $|R|$ & -- & -- & $\emptyset$ & $\hat{r}$ & FPT \\
  \cmidrule{1-4}
  -- & -- & \begin{tabular}{@{}c@{}}$|P(r)| \le 3$\\ $1 \le |P(r) \cap P_{\ell b}| \le 3$ \\ $|P(r) \cap P(r')| \le 2$\end{tabular} & $\emptyset$ & $|P_{\ell b}| + k_p$ & W[1]-hard \\
  \cmidrule{1-4}
  -- & -- & $K_{2,2}$-free & \begin{tabular}{@{}c@{}}$\{\langle X_1,2\rangle,\dots,\langle X_m,2\rangle\}$\\ $|X_i|=2$ \end{tabular} & $k_r + \hat{k}$ & W[1]-hard \\
  \cmidrule{1-4}
  -- & -- & $K_{\alpha,\beta}$-free & \begin{tabular}{@{}c@{}} $\{\langle X_1,t_1\rangle,\dots,\langle X_m,t_m\rangle\}$ \\ $|X_i| \le c$ \\ $X_i \cap X_j = \emptyset$ \end{tabular} & $k_r + \hat{k}$ & FPT \\
  \bottomrule
 \end{tabular}
\end{table}

\section{Preliminaries}\label{sec:prel}

\paragraph{Terminology and Notation.}
For a graph $G = \langle V(G), E(G) \rangle$ and vertex $x\in V(G)$, $N_G(x) = \{y \in V(G) \mid xy \in E(G)\}$ is the set of vertices adjacent to $x$ (``neighbors''). 
For a set $S\subseteq V(G)$, $N(S)=\bigcup_{x\in S} N_G(x) \setminus S.$
We will omit the subscript $G$ when the graph is clear from the context.

A graph $G$ is {\em bipartite} if its vertices can be partitioned into two sets $A$ and $B$ such that for all $ab \in E$, $a \in A$ and $b \in B$.
We will generally write a bipartite graph $G$ in the form $G=(A\uplus B,E)$.
A bipartite graph $G=(A\uplus B,E)$ is {\em complete} if $ab\in E$ for every $a\in A$ and $b\in B$.
Such a graph will also be denoted (up to isomorphism) by $K_{\alpha,\beta}$, where $\alpha=|A|$ and $\beta=|B|$.
{A graph is $K_{\alpha,\beta}$-{\em free} if it contains no induced subgraph isomorphic to $K_{\alpha,\beta}$.}

For a positive integer $k$, we write $[k]$ to denote $\{1,2,\dots,k\}$. 


\paragraph{Parameterized Complexity.} An instance of a parameterized problem $\Pi$
is a pair $(I,k)$ where $I$ is the {\em main part} and $k$ is the
{\em parameter}; the latter is usually a non-negative integer.  
A parameterized problem is
{\em fixed-parameter tractable} ({\sf FPT}) if there exists a computable function
$f$ such that instances $(I,k)$ can be solved in time $O(f(k)|{I}|^c)$
where $|I|$ denotes the size of~$I$ and $c$ is an absolute constant. The class of all fixed-parameter
tractable decision problems is called {{\sf FPT}} and algorithms which run in
the time specified above are called {{\sf FPT}} algorithms. As in other literature on {{\sf FPT}} algorithms,
we will often omit the polynomial factor in $\cO(f(k)|{I}|^c)$ and write $\cO^*(f(k))$ instead.

Consider two parameterized problems $\Pi$ and $\Pi'$. We say that $\Pi$ has a {\em parameterized reduction} to $\Pi'$ if there are functions $k\mapsto k'$ and $k\mapsto k''$ from $\mathbb{N}$ to $\mathbb{N}$
and a function $(I,k)\mapsto (I',k')$ such that 

\begin{enumerate}
 \item  $(I,k)\mapsto (I',k')$ is computable in $k''(|I|+k)^{O(1)}$ time, and 
 \item $(I,k)$ is a yes-instance of $\Pi$ if and only if $(I',k')$ is a yes-instance of $\Pi'$.
\end{enumerate}

While {\sf FPT} is a parameterized complexity analog of {\sf P} in classic complexity theory, there are many parameterized hardness classes, forming a nested sequence of which {\sf FPT} is the first member: {\sf FPT}$\subseteq$ {\sf W}[1]$\subseteq$ {\sf W}[2 $]\subseteq \dots$.
It is well known that if the Exponential Time Hypothesis holds then ${\sf FPT} \ne {\sf W}[1]$.\footnote{The Exponential Time Hypothesis is a conjecture that there is no algorithm solving 3-CNF Satisfiability in time $2^{o(n)}$, where $n$ is the number of variables.}
Hence, {\sf W}[1] is generally viewed as a parameterized intractability class, which is an analog of {\sf NP} in classical complexity.  {Consider the following two parameterized problems. In the {\sc Clique} problem parameterized by $k$, given a graph $G$ and a natural number $k$, we are to decide whether $G$ has a complete subgraph on $k$ vertices. In the {\sc Dominating Set} problem parameterized by $k$, given a graph $G=(V,E)$ and a natural number $k$, we are to decide whether $G$ has a set $S$ of vertices such that every vertex in $V\setminus S$ is adjacent to some vertex in $S.$
A parameterized problem $\Pi$ is in {\sf W}[1] ({\sf W}[2], respectively) if it there is parameterized reduction from $\Pi$ to {\sc Clique} ({\sc Dominating Set}, respectively).}
Thus, every {\sf W}[1]-hard problem $\Pi_1$ ( {\sf W}[2]-hard problem $\Pi_2$, respectively) is not `easier' than {\sc Clique} ({\sc Dominating Set}, respectively), i.e., {\sc Clique} ({\sc Dominating Set}, respectively) has a parameterized reduction to $\Pi_1$ ($\Pi_2,$ respectively). 

More information on parameterized algorithms and complexity can be found in recent books~\cite{cygan2015,downey2013}.

\section{Parameterized Hardness of UAQs without SoD constraints}\label{sec:hard}


In this section, we consider {\UAQ} problems with $\cD=\emptyset$.
We will first consider two {\sc Simple} {\UAQ} problems defined as follows:

\begin{description}
\item[{\sc Simple UAQ of Type 1}:]  $k_r = |R|$ and $\cD = \emptyset$. Thus, an instance of {\sc $k_p$-Simple} UAQ problem can be written as a tuple $(R, P, RP, P_{\ell b}, k_p)$.
\item[{\sc Simple UAQ of Type 2}:] $P_{\ell b}  = P$ and $\cD = \emptyset$. Thus, an instance of {\sc $k_r$-Simple} UAQ problem can be written as a tuple $(R, P, RP, k_r)$.
\end{description}  

Note that {\sc Simple UAQ of Type 1} is a natural simplification of UAQ, which prior research has established is {\sf NP}-hard~\cite{ChCr09}, thus establishing that more complex variants of the problem are also hard.
Below we will strengthen this result, showing (i)~that it is ``hard'' from the {\sf FPT} perspective (unlike many versions of WSP~\cite{CoCrGaGuJo14}), and (ii)~that {\sf FPT} versions of the problem do exist when a different small parameter is considered. 
Type 2 problems essentially ask whether it is possible to find a small set of roles that are collectively authorized for a set of permissions. 
An algorithm to solve this problem may well have use in applications such as role mining.

\begin{theorem}
\label{thm:hardness-simplified-uaq}
The {\sc Simple  UAQ problem of Type 1} ({\sc of Type 2}, respectively) parameterized by $k_p$ (by $k_r$, respectively) is {\sf W}[2]-hard.
\end{theorem}

To prove this theorem, we will use the following problem, which is {\sf W}[2]-complete \cite{DowneyF99}.

\defparproblem{\RBDS}{A bipartite graph $G = (A\uplus B, E)$}{$k$}{Is there a subset $S$ of $A$ of size $k$ such that $N(S)=B$?}

The proof of the theorem is based on parameterized reductions from {\RBDS} to the Simple UAQ problems.
 
\begin{proof}[Proof of Theorem~\ref{thm:hardness-simplified-uaq}] Let $(G=(A\uplus B, E), k)$ be an instance of {\RBDS} problem.

\vspace{2mm}

{\bf Type 1.} Let $L = \{p_v \mid \ v \in A\}.$ Set $R = A$, $P_{\ell b}=B,$
$P = B\uplus L,$ $RP=E\cup \{(v,p_v) \mid \ v \in A\}$  and $k_p = k$. Let $S\subseteq A.$
Note that by definitions of \RBDS{} and Simple  UAQ problem of Type 1, $S$ is a solution of an {\RBDS} instance, i.e. $N(S)=B$ and $|S|\le k,$ if and only if
$S$ is a solution of the corresponding {UAQ} instance, i.e., $P_{lb}\subseteq P(S)$ and $|P(S)\setminus P_{lb}|=|{p_v: v\in S}|\le k = k_p.$
Thus, $(G=(A\uplus B, E), k)$ is a yes-instance of {\RBDS} if and only if $(R, P, RP, P_{\ell b}, k_p)$ is a yes-instance of {\sc Simple UAQ of Type 1}. 
Since  {\RBDS} is {\sf W}[2]-hard this reduction shows that {\sc Simple UAQ of Type 1} is also {\sf W}[2]-hard.

\vspace{2mm}

{\bf Type 2.} 
Set $R = A$, $P = B,$ $RP=E$  and $k_r = k$.
This gives us an instance $(R, P, RP, k_r)$ of {\sc Simple UAQ of Type 2}.  
By the definitions of {\RBDS} and {\sc Simple UAQ of Type 2} and setting $R_s=S$, it follows that $(G=(A\uplus B, E), k)$ is a yes-instance of {\RBDS} if and only if $(R, P, RP, k_r)$ is a yes-instance of {\sc Simple UAQ of Type 2}. 
Since  {\RBDS} is {\sf W}[2]-hard this reduction shows that {\sc Simple UAQ of Type 2} is also {\sf W}[2]-hard.
\end{proof}

Theorem~\ref{thm:hardness-simplified-uaq} asserts that {\sc Simple UAQ of Type 1} parameterized by $k_p$ is {\sf W}[2]-hard, implying that it is highly unlikely to be {\sf FPT}. 
This result is somewhat unexpected, given that WSP (which appears to be a more complex problem) is FTP for most instances that are likely to arise in practice~\cite{CoCrGaGuJo14,CrGaGuJoWa16}.

However, we are able to prove that {\sc Simple UAQ of Type 1} is {\sf FPT} when parameterized by a different parameter.  
Let $R_1 = \{r \in R \mid P(r) \subseteq P_{lb}\}$, $R_2 = R \setminus R_1$, and $\hat{r}=|R_2|.$
In other words, $\hat{r}$ is the number of roles that are assigned to at least one permission outside $P_{lb}$.
We then have the following result.

\begin{theorem}
{\sc Simple UAQ of Type 1}  admits an algorithm of running time $\cO^*(2^{\hat{r}}).$ Thus, {\sc Simple UAQ of Type 1}  parameterized by $\hat{r}$ is {\sf FPT}.
\end{theorem}

\begin{proof}
Let ${\cI}$ be an instance of {\sc Simple UAQ of Type 1}. 
We consider every subset $S_2$ of $R_2$.
If $|P(S_2 \cup R_1)| \leq k_p + |P_{lb}|$ and $P(S_2 \cup R_1) \supseteq P_{lb}$, then ${\cI}$ is a yes-instance.
If no such $S_2$ exists, then ${\cI}$ is a no-instance.

Checking whether $|P(S_2 \cup R_1)| \leq k_p + |P_{lb}|$ and $P(S_2\cup R_1) \supseteq P_{lb}$ can be done in polynomial time.  
Thus, the overall running time is $\cO^*(2^{\hat{r}})$.
\end{proof}

It is not clear how useful this result will be in practice.
Further research is required to determine the likelihood of this parameter being small in real-world instances.

In the rest of this section, we will prove the following:

\begin{theorem}\label{thm:noSoD2}
{UAQ} parameterized by $|P_{\ell b}| + k_p$ is {\sf W}[1]-hard even under the following restrictions:
\begin{itemize}
	\item for all $r \in R$, $|P(r)| \leq 3$ and $1 \leq |P(r) \cap P_{lb}| \leq 3$ , and
	\item for all $r,r' \in R$, $|P(r) \cap P(r')| \leq 2$.
\end{itemize} 
\end{theorem}

The first restriction requires that every role is authorized for at most three permissions and at most two permissions not in $P_{lb}$. The second restriction requires that every pair of roles is authorized for at most two common permissions.

It is worth noting that the first restriction is unlikely to be satisfied by real-world RBAC instances. RBAC is based on an assumption that the complexity of managing access control systems can be reduced by introducing 
(a relatively small number of) roles as an abstraction acting as an intermediate layer between users and permissions and the roles will be assigned to many permissions and many users. 
(In other words the number of relationships that needs to be maintained is $O(|R|(|U| + |P|))$ rather than $O(|U||P|)$.) In short, imposing such a small upper bound on the number of permission assignments to each role means that this result is unlikely to be applicable to many real-world instances. (The second restriction is less problematic, at least in the non-hierarchical setting, as one would expect sets of permissions assigned to different roles to be approximately disjoint, with users being assigned to several roles.)

To prove this theorem, we provide a parameterized reduction from {\mulbiclique}, which is known to be {\sf W}[1]-hard~\cite{cygan2015}, defined as follows.

\defparproblem{\mulbiclique}{A bipartite graph $G = (A\uplus B, E), A = A_1 \uplus \cdots \uplus A_k, B = B_1 \uplus \cdots \uplus B_k$}{$k$}{Is there $A' \subseteq A, B' \subseteq B$ such that for every $i \in [k]$, $|A_i \cap A'| = |B_i \cap B'| = 1$, and $A' \cup B'$ induces a complete bipartite graph?}

\begin{construction}\label{construction:uaq-from-multi-biclique}
Consider an instance of {\mulbiclique}, where $G = (A\uplus B, E)$ is a bipartite graph, $A = A_1 \uplus A_2 \uplus \cdots \uplus A_k$, and $B = B_1 \uplus B_2 \uplus \ldots \uplus B_k$.
We construct an instance of {\UAQ}  as follows:
        \begin{itemize}
	\item $R = \{r_{uv}  \mid   uv \in E\}$,
	\item $P_{\ell b} = \{p_{i, j}  \mid  (i, j) \in [k] \times [k]\}$, 
	\item $P = P_{\ell b} \cup A \cup B$,
	\item for every $uv \in E$, if $u \in A_i, v \in B_j$, then $P(r_{uv}) = \{u, v, p_{i,j}\}$,
	\item $k_r = k^2$, $k_p = 2k.$
	\end{itemize}
\end{construction}

In the construction we create a role for every edge of $G$, a permission for every vertex in $G$, and a permission for every ordered pair in $[k]\times [k]$.
Moreover, the vertices of $G$ form $P\setminus P_{lb}.$ 
The intuition behind the construction is that the permissions associated with a role $r_{e}$ encode an edge $e = uv \in E$, and the pair $(A_i, B_j)$ to which $u$ and $v$ respectively belong.
Note that the {\UAQ} instance generated by Construction~\ref{construction:uaq-from-multi-biclique} satisfies the criteria in Theorem~\ref{thm:noSoD2}. 

Consider the instance of {\sc Multicolored Biclique} shown in Figure~\ref{fig:c1}(a) comprising sets $A = \{a_1,a_2\} \uplus \{a_3,a_4,a_5\}$ and $B = \{b_1,b_2\} \uplus \{b_3,b_4\}$. The instance has a solution, also shown in Figure~\ref{fig:c1}(a). Construction 1 generates a UAQ instance with solution $\{r_{22},r_{23},r_{32},r_{33}\}$, corresponding to the edges in the biclique solution. The permissions associated with these roles respectively include $p_{11}$, $p_{12}$, $p_{21}$ and $p_{22}$ (the set $P_{\ell b}$), reflecting the fact that the {\sc Multicolored Biclique} solution must contain an edge from every block in the partition of $A$ to every block in the partition of $B$. (Hence, there must be $k^2$ such permissions - four in this example.) Moreover, each role is assigned two further permissions, corresponding to the endpoints of the edges defining the roles. (Hence, there are a further $2k$ permissions - four in this example.)  Note that in Figure~\ref{fig:c1}(b) we
used dashed and dotted lines in the UAQ solution to differentiate between the edges that "encode" the required relationships between the blocks in the  {\sc Multicolored Biclique} instance and the edges in the  {\sc Multicolored Biclique} solution, respectively.

\begin{figure}[h]
 \begin{subfigure}[b]{0.6\textwidth}\centering
 \begin{tikzpicture}[node distance=1cm and 2cm,inner sep=1pt,minimum size=6pt,scale=0.75, transform shape]
  \node[circle,draw,fill,label=left:$a_1$] (a1) {};
  \node[circle,draw,fill,below=of a1,label=left:$a_2$] (a2) {};
  \node[circle,draw,fill=white,below=of a2,label=left:$a_3$] (a3) {};
  \node[circle,draw,fill=white,below=of a3,label=left:$a_4$] (a4) {};
  \node[circle,draw,fill=white,below=of a4,label=left:$a_5$] (a5) {};
  \node[circle,draw,fill,right=of a1,label=right:$b_1$] (b1) {};
  \node[circle,draw,fill,below=of b1,label=right:$b_2$] (b2) {};
  \node[circle,draw,fill=white,below=of b2,label=right:$b_3$] (b3) {};
  \node[circle,draw,fill=white,below=of b3,label=right:$b_4$] (b4) {};
  \draw (a1) to (b1);
  \draw (a1) to (b2);
  \draw (a1) to (b4);
  \draw (a2) to (b2);
  \draw (a2) to (b3);
  \draw (a3) to (b2);
  \draw (a3) to (b3);
  \draw (a4) to (b1);
  \draw (a4) to (b4);
  \draw (a5) to (b1);
  \draw (a5) to (b4);
  \node[circle,draw,fill,right=of b2,label=left:$a_2$] (a2s) {};
  \node[circle,draw,fill=white,below=of a2s,label=left:$a_3$] (a3s) {};
  \node[circle,draw,fill,right=of a2s,label=right:$b_2$] (b2s) {};
  \node[circle,draw,fill=white,below=of b2s,label=right:$b_3$] (b3s) {};
  \draw (a2s) to (b2s);
  \draw (a2s) to (b3s);
  \draw (a3s) to (b2s);
  \draw (a3s) to (b3s);
 \end{tikzpicture}
 \caption{{\sc Multicolored Biclique} instance and solution}
 \end{subfigure}
 \hfill
 \begin{subfigure}[b]{0.35\textwidth}\centering
 \begin{tikzpicture}[node distance=1cm and 2cm,inner sep=1pt,minimum size=6pt,scale=0.75, transform shape]
  \node[circle,draw,fill,label=left:$r_{22}$] (r22) {};
  \node[circle,draw,fill,below=of r22,label=left:$r_{23}$] (r23) {};
  \node[circle,draw,fill,below=of r23,label=left:$r_{32}$] (r32) {};
  \node[circle,draw,fill,below=of r32,label=left:$r_{33}$] (r33) {};
  \node[circle,draw,fill,right=of r22,label=right:$p_{21}$] (p21) {};
  \node[circle,draw,fill,above=of p21,label=right:$p_{12}$] (p12) {};
  \node[circle,draw,fill,above=of p12,label=right:$p_{11}$] (p11) {};
  \node[circle,draw,fill,below=of p21,label=right:$p_{22}$] (p22) {};
  \node[circle,draw,fill,below=of p22,label=right:$a_2$] (a2) {};
  \node[circle,draw,fill,below=of a2,label=right:$a_3$] (a3) {};
  \node[circle,draw,fill,below=of a3,label=right:$b_2$] (b2) {};
  \node[circle,draw,fill,below=of b2,label=right:$b_3$] (b3) {};
  \draw[dashed] (r22) to (p11);
  \draw[dashed] (r23) to (p12);
  \draw[dashed] (r32) to (p21);
  \draw[dashed] (r33) to (p22);
  \draw[dotted] (r22) to (a2);
  \draw[dotted] (r22) to (b2);
  \draw[dotted] (r23) to (a2);
  \draw[dotted] (r23) to (b3);
  \draw[dotted] (r32) to (a3);
  \draw[dotted] (r32) to (a2);
  \draw[dotted] (r33) to (a3);
  \draw[dotted] (r33) to (b3);
 \end{tikzpicture}
 \caption{UAQ solution}
 \end{subfigure}
 \caption{{\sc Multicolored Biclique} instance and corresponding UAQ solution using Construction 1}\label{fig:c1}
\end{figure}
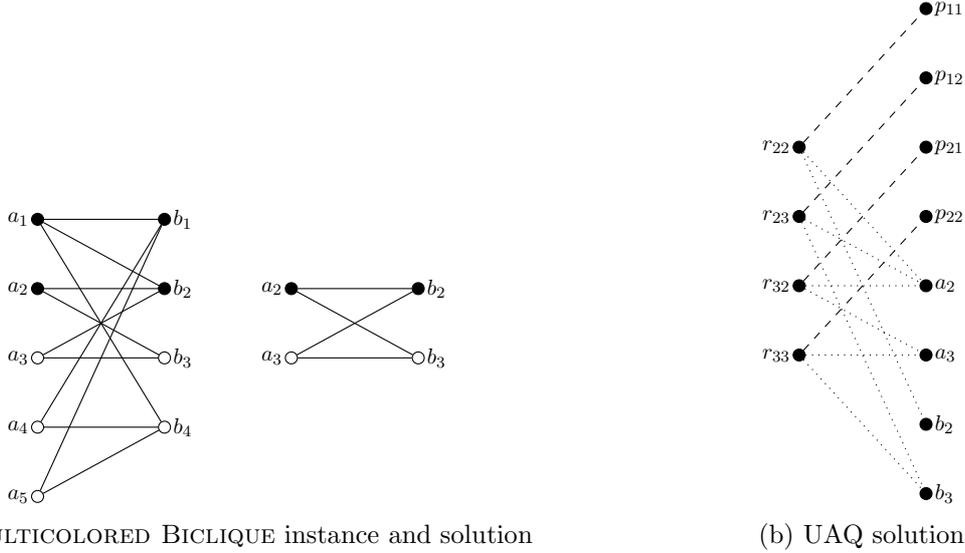

We have the following:

\begin{lemma}
\label{lemma:equivalece-mult-biclique-uaq-3}
 Let $(R, P, RP, P_{\ell b},  \emptyset, k_r, k_p)$ be a {UAQ} instance created from an instance $(G, k)$ of {\mulbiclique}.
Then, $(G, k)$ is a yes-instance if and only if $(R, P, RP, P_{\ell b},  \emptyset, k_r, k_p)$ is.
\end{lemma}

\begin{proof}
First we give the forward direction $(\Rightarrow)$ of the proof.
Let $(G, k)$ be a yes-instance of {\mulbiclique} and suppose that $S \subseteq V$ be a solution of $(G, k)$.
This means that for every $i \in [k]$, $|S \cap A_i| = |S \cap B_i| = 1$.
We construct the role set $R_s \subseteq R$ using $S \subseteq V$ as follows.
We put $r_{uv}$ into $R_s$ if and only if $uv$ is an edge in $G[S]$.
Observe that $P(R_s) \supseteq P_{\ell b}$ since $S$ intersects every $A_i$ and every $B_j$.
Since every role of $R_s$ corresponds to an edge in $G[S]$, and there are $k^2$ edges in $G[S]$, we have $|R_s| = k^2$.
Since there are $2k$ vertices in $S$, $R_s$ is authorized for $2k$ permissions in $P \setminus P_{\ell b}$.

Now, we give the backward direction $(\Leftarrow)$ of the proof.
Let $R_s \subseteq R$ be a set of at most $k^2$ roles that is authorized for all permissions in $P_{\ell b}$ and at most $2k$ permissions in $P \setminus P_{\ell b}$.
Consider the set $S$ of permissions from $P \setminus P_{\ell b}$ authorized by $R_s$.
We claim that $G[S]$ is a complete bipartite graph such that for all $i \in [k], |S \cap A_i| = |S \cap B_i| = 1$.
Consider an arbitrary permission $p_{i,j} \in P_{\ell b}$.
Only a role $r_{uv}$ such that $u \in A_i, v \in B_j$ is authorized for this permission.
Hence, for every $i \in [k]$, $|S \cap A_i|, |S \cap B_i| \geq 1$.
Furthermore, the roles in $R_s$ are assigned at most $2k$ permissions from $P \setminus P_{\ell b}$.
So, $|S \cap A_i| = |S \cap B_i| = 1.$
For any role $r_{uv} \in R_s$, $uv \in E$.
So, $G[S]$ induces a complete bipartite graph.
Hence, $(G, k)$ is a yes-instance for {\mulbiclique}. 
\end{proof}


\noindent{\bf Proof of Theorem \ref{thm:noSoD2}.} 
By Construction~\ref{construction:uaq-from-multi-biclique}, there exists a polynomial time reduction from {\mulbiclique} parameterized by $k$ to {UAQ} parameterized by $k^2 + 2k$ where $k_r = k^2, k_p = 2k$.
Also by this construction, every role is authorized for at most two permissions from $P \setminus P_{\ell b}$ and at most three permissions from $P$.
Since every role is authorized for exactly three permissions and no pair of roles is authorized for the same set of permissions, every pair of roles is authorized for at most two common permissions.
By construction, $|P_{\ell b}| = k^2$.
Hence, by Lemma \ref{lemma:equivalece-mult-biclique-uaq-3} and the fact that {\sc Multicolored Biclique} is W[1]-hard, {UAQ} parameterized by $|P_{\ell b}| + k_p$ is {\sf W}[1]-hard even when every role is authorized at most three permissions from $P$ 
and at most two permissions from $P \setminus P_{\ell b}$ and every two roles are authorized for at most two common permissions.
\qed

\section{W[1]-hardness when RPG is $K_{2,2}$-free}
\label{sec:w-hardness-k-ij}


We now consider restricting UAQ to instances in which $|P(r) \cap P(r')| \leqslant 1$ for all $r,r' \in R$.
In other words, any two roles are authorized for at most one common permission.
Technically, this is equivalent to saying the bipartite graph RPG contains no subgraph isomorphic to $K_{2,2}$.
Hence we will call such instances of UAQ \emph{$K_{2,2}$-free}. We show that UAQ for $K_{2,2}$-free instances is W[1]-hard. We believe that this result is useful because it demonstrates that UAQ remains a hard problem 
even when we impose a strong restriction (and one that few RBAC instances are likely to satisfy) on RPG. The practical consequence of this result is that UAQ is hard for most real-world RBAC instances with constraints. 

In Section 5, we relax the condition on RPG somewhat (making UAQ $K_{\alpha,\beta}$-free for $\alpha,\beta > 2$) but impose restrictions on the set of constraints. We show that UAQ is FPT with these restrictions. 
This positive result means that UAQ in real-world RBAC instances \emph{without} constraints (i.e., instances that trivially satisfy any restrictions on the set of constraints) may be solvable in a reasonable amount of time. 
However, the FPT algorithm required to solve UAQ is extremely complicated and requires sophisticated programming techniques. Moreover, the non-polynomial terms in the running time are significant. In short, 
we believe it is unlikely that a simple and practical FPT algorithm exists for instances of this type. 

In the remainder of the paper we will write ${\hat k}$ to denote the number of permissions that are not in $P_{\ell b}$ (i.e., $|P \setminus P_{\ell b}|$).

\begin{theorem}
\label{theorem:hardness-uaq-pairwise-intersection}
Consider a $K_{2,2}$-free {UAQ} instance.
Then, {UAQ} parameterized by $k_r +{\hat k}$  is {\sf W}[1]-hard even when every SoD constraint $\ev{X,t}$ in $\cD$ has $|X| = t = 2$.	
\end{theorem}

The proof is based on a parameter-preserving reduction from {\sc Multicolored Biclique} to {UAQ}.
We construct a {UAQ} instance as follows. (This construction is similar to the one used in~\cite{JacobM019}.)

\begin{construction}\label{construction:multi-biclique-to-uaq-k22-free}
Given $G = (V,E)$, where $V = A_1 \uplus \dots \uplus A_k \uplus B_1 \uplus \dots \uplus B_k$: 
\begin{itemize}
      \item $R = \{r_v \mid  v \in V\} \cup \{s\}$;
	\item $P = \{p_v \mid  v \in V\} \cup [2k] \cup \{q\}$, and $ P_{\ell b} = P$;
	\item for every $i \in [k]$, if $v \in A_i$, then $P(r_{v}) = \{p_v, i\}$;
    \item for every $i \in [2k] \setminus [k]$, if $v \in B_{i-k}$, then $P(r_{v}) = \{p_v, i\}$;
	\item $P(s) = \{p_v \mid  v \in V\} \cup \{q\}$;
	\item given $u \in A$, $v \in B$, $\langle \{r_u,r_v\} , 2 \rangle \in \cD$ iff $uv \notin E$;
	\item $k_r = 2k + 1$ and $k_p = 0$.
\end{itemize}
\end{construction}

Construction \ref{construction:multi-biclique-to-uaq-k22-free}
defines a role and a permission for each vertex. Additionally, we define role $s$ and permissions $1,\ldots ,2k$ and $q$. The additional role and permissions are used to encode the structure of the bipartite graph in terms of 
the sets into which $A$ and $B$ are partitioned, via the $RP$ relation. 
Finally, we use separation-of-duty constraints to ensure that we prohibit solutions to UAQ that are not consistent with the structure of the bipartite graph in the {\sc Multicolored Biclique} instance.

Observe that the resulting {UAQ} instance is non-hierarchical.
We now show it is $K_{2,2}$-free. 

\begin{lemma}
\label{lemma:instance-has-pairwise-intersection}
Let $(G = (A\uplus B, E), k)$ be an instance of {\mulbiclique}.
Then the UAQ instance derived using Construction~\ref{construction:multi-biclique-to-uaq-k22-free} is $K_{2,2}$-free.
\end{lemma}

\begin{proof}
By construction, $R = \{r_v  \mid  v \in V\} \cup \{s\}$.
Consider $r_u, r_v$ for two distinct vertices $u, v \in V$.
The following cases can arise.
\begin{description}
	\item[Case 1:] If $u, v \in A_i$ for some $i \in [k]$, then $P(r_u) = \{p_u, i\}, P(r_v) = \{p_v, i\}$.
	Similarly, if $u, v \in B_{i-k}$ for some $i \in [2k] \setminus [k]$, then $P(r_u) = \{p_u, i\}, P(r_v) = \{p_v, i\}$.
	Then, $|P(r_u) \cap P(r_v)| \leq 1$.
	
	\item[Case 2:] If $u \in A_i, v \in A_j$ for $i, j \in [k], i \neq j$, then $P(r_u) \cap P(r_v) = \{p_u, i\} \cap \{p_v, j\}= \emptyset$.
	Similarly, if $u \in B_{i - k}, v \in B_{j - k}$ for two distinct $i, j \in [2k] \setminus [k]$, then $P(r_u) \cap P(r_v) = \emptyset$.
	Similarly, if $u \in A_i, v \in B_{j - k}$ for $i \in [k], j \in [2k] \setminus [k]$, then $P(r_u) \cap P(r_v) = \emptyset$.
	
	\item[Case 3:] Let $u \in V$.
	Recall that $P(s) = \{p_u \mid u \in V\} \cup \{q\}$, and $P(r_u) = \{p_u, i\}$ for some $i \in [2k]$.
	Hence, $P(r_u) \cap P(s) = \{p_u\}$.  
\end{description}
The above cases are exhaustive and, in each case, every pair of distinct roles is assigned at most one common permission.
\end{proof}

\begin{lemma}
\label{lemma:hardness-kij-free}
$(R, P, RP, \cD,k_r,k_p)$ is a yes-instance for {UAQ} if and only if $(G,k)$ is a yes-instance for {\mulbiclique}.
\end{lemma}

\begin{proof}
First we give the forward direction $(\Rightarrow)$ of the proof.
Let $(G, k)$ be a yes-instance of {\mulbiclique}.
Then, there exists $S \subseteq V$ such that $|S| = 2k$, for all $i \in [k]$, $|A_i \cap S| = |B_i \cap S| = 1$, and $G[S]$ induces a complete bipartite graph.
We construct $R^*$ from $S$ as follows:
$R^* = \{r_u \mid  u \in S\} \cup \{s\}$.
By construction $|R^*| = 2k + 1 = k_r$.
Since for every $i \in [k]$, $S$ contains exactly one vertex from $A_i$, $R^*$ is authorized for all permissions in $[k]$.
Since for every $i \in [2k] \setminus [k]$, $S$ contains exactly one vertex from $B_{i - j}$, $R^*$ is authorized for all permissions in $[2k] \setminus [k]$.
The role $s$ is authorized for all permissions in $P \setminus [2k]$.
Hence, $R^*$ is authorized for all permissions in $P.$ 

Consider an arbitrary SoD constraint $\ev{\{r_u, r_v\},2} \in \cD$.
By construction, we have $u \in A, v \in B$, and $uv \notin E$.
But, then $\{u, v\} \not\subseteq S$.
However, by construction, $G[S]$ induces a complete bipartite graph in $G$ with bipartition $S_1 \uplus S_2$ such that $S_1 \subseteq A, S_2 \subseteq B$.
Hence, for every $u \in S_1, v \in S_2$, $uv \in E$.
Thus, all the SoD constraints are satisfied for $R^*$.
Hence, $(R, P, RP, \cD, k_r, k_p)$ is a yes-instance of {UAQ}.

Now, we give the backward direction $(\Leftarrow)$ of the proof.
Let $(R, P, RP, \cD, k_r, k_p)$ be a yes-instance of {UAQ}.
Then, there exists a set $R^*$ of at most $k_r$ roles that are authorized for all permissions of $P.$ 
Observe that $s \in R^*$ since $q \in P$.
Moreover, $R^*$ is authorized for all permissions in $[2k]$.
Every $i \in [k]$ can be authorized by a unique role $r_v$ such that $v \in A_i$.
Similarly, every $i \in [2k] \setminus [k]$ can be authorized by a unique role $r_v$ such that $v \in B_{i - k}$.
We construct $S$ from $R^*$ as follows.
We put $v \in S$ if $r_v \in R^*$.
Since $R^* \setminus \{s\}$ contains $2k$ roles, $S$ contains $2k$ vertices.
Hence, $|S| = 2k$ and $|S \cap A_i| = |S \cap B_i| = 1$.

We have $S = S_1 \uplus S_2$, where $S_1=S\cap A$ and $S_2=S\cap B.$
Note that $R^*$ satisfies all constraints in $\cD$ and consider a constraint $\ev{\{r_u,r_v\}, 2} \in \cD$ such that w.l.o.g. $u\in A$ and $v\in B$. 
Since $R^*$ satisfies this constraint, $|R^*\cap \{r_u,r_v\}|\le 1$ and hence $|\{u,v\}\cap S|\le 1$. Hence, for every $u\in S_1$ and $v\in S_2$, we have $uv\in E$ 
implying that $G[S]$ is a complete bipartite graph.
 Therefore, $(G, k)$ is a yes-instance for {\mulbiclique}.  
\end{proof}

\begin{proof}[Proof of Theorem~\ref{theorem:hardness-uaq-pairwise-intersection}]
Based on the above-mentioned construction and Lemma~\ref{lemma:hardness-kij-free}, an instance $(G, k)$ of {\mulbiclique} can be transformed into an equivalent instance $(R, P, RP, \cD,k_r,k_p)$ of {UAQ} where every SoD constraint is of the form $\ev{X, 2}$ such that $|X| = 2$.
By construction, $P_{lb} = P$.
The parameter is transformed from $k$ to $k_r + {\hat k} = 2k + 1$.
Hence, this is a parameterized reduction and the fact that {\sc Multicolored Biclique} is W[1]-hard, {UAQ} is {\sf W}[1]-hard when parameterized by $k_r + {\hat k}$.
\end{proof}

Let $\alpha\ge 2$ and $\beta\ge 2$ be fixed integers and assume that every $\alpha$ roles  are authorized for at most $\beta-1$ common permissions. In other words, RPG is $K_{\alpha,\beta}$-free.  We call such UAQ the {\em $K_{\alpha,\beta}$-free UAQ} problem. 
Note that if RPG is  $K_{2,2}$-free then it is $K_{\alpha,\beta}$-free for every $\alpha\ge 2$ and $\beta\ge 2.$
Thus, $K_{2,2}$-free UAQ is a special case of $K_{\alpha,\beta}$-free UAQ. Therefore, parameterized intractability of Theorem~\ref{theorem:hardness-uaq-pairwise-intersection} can be extended to arbitrary   integers $\alpha\ge 2$ and $\beta\ge 2$. As the main result of the next section shows, this parameterized intractability can be attributed to the fact that here SoD constraints can have overlapping role sets.

\section{Fixed-Parameter Tractability when RPG is $K_{\alpha,\beta}$-free and Constraints are Non-intersecting}
\label{sec:k-i-j-free}

Let $\alpha\ge 2$ and $\beta\ge 2$ be fixed integers.
In this section, we consider the {\UAQ} problem restricted by the following conditions:
\begin{description}
\item[(i)] for every set of $\alpha$ roles, $\{r_1,\ldots,r_{\alpha}\}$, $|P(r_1) \cap P(r_2) \cap \cdots \cap P(r_{\alpha})| < \beta$;
\item[(ii)] there is a constant $c$ such that for every SoD constraint $\ev{X, t}$ we have $|X|\le c$;
\item[(iii)] for every pair $\ev{X_1, t_1}, \ev{X_2, t_2}$ of SoD constraints, we have  $X_1 \cap X_2 = \emptyset$.
\end{description}
We call this the $(\alpha, \beta)$-{\UAQ} problem ($(\alpha, \beta)$-{UAQ}). 
Note that the RPG of an instance of $(\alpha, \beta)$-{UAQ} is $K_{\alpha, \beta}$-free. 

Recall that $\hat k=|P \setminus P_{\ell b}|$. 
We consider $(\alpha, \beta)$-{UAQ} parameterized by $k_r + {\hat k}$ and prove that it is {\sf FPT} by  designing an algorithm with running time $\cO^*(2^{\cO(k_r^{\alpha} + \hat k)})$ to solve the problem. 
The algorithm makes use of matroids and dynamic programming.
Moreover, we have to perform some preprocessing on an $(\alpha,\beta)$-{UAQ} instance to produce the input to the algorithm. {Note that the algorithm is not intended to inspire implementations, but rather offer a constructive proof of the main result in this section.}

In Section~\ref{sec:red}, we use reduction and branching rules to reduce the size of the original instance. 
In particular, this preprocessing phase reduces the size of $P_{\ell b}$ to $\cO(\beta k_r^{\alpha})$.
This preprocessing takes $\cO^*(\alpha^{k_r})$ time and also computes a partial solution $R_1$ to the input instance.  
In every reduction and branching rule, if we delete a role $r$ from the input and reduce $k_r$, then we add $r$ to $R_1$. Otherwise, if $k_r$ remains unchanged, we just delete $r$ from the input instance. 

After the preprocessing phase, it is possible to use a dynamic programming algorithm to determine whether there \emph{exists} a solution to the original instance.
We can use the same algorithm to determine the \emph{size} of the solution set.
However, more advanced techniques are required to compute a \emph{solution set}.

Recall that a dynamic programming (DP) algorithm stores and re-uses solutions to smaller sub-problems.
In order to find a solution to $(\alpha,\beta)$-{UAQ}, we need to store in the DP table some candidate partial solutions in each of the DP table entries. 
However, it is not sufficient to store one partial candidate solution in each DP table entry since we do not know which specific partial candidate solution would extend to a candidate solution for the entire instance. 

A naive approach could store all possible candidate partial solutions in the DP table entries. 
Then the number of candidate partial solutions in a  DP table entry could be as large as $|R|^{k_r}$; so this approach would not result in an {\sf FPT} algorithm. 
Thus, we have to store only a ``small'' subset of partial candidate solutions in each DP table entry. 
To be able to do this, we use an advanced algorithmic technique known as the \emph{method of representative families} on matroids.
This method  ensures that if the input instance is a yes-instance, then a solution will be found, despite storing only $\cO^*(2^{k_r})$ partial candidate solutions in a table entry, stored in a table containining $\cO^*(2^{\cO(k_r^{\alpha}) + \hat k})$ entries. 

In Section~\ref{sec:pdp}, we provide the terminology, notation and results on matroids and representative families that are necessary for describing and analyzing our DP algorithm based on the method of representative families on matroids. 
We also describe how to construct a special partition matroid used in the DP algorithm.
This matroid enables us to identify those candidate solutions that do not violate any constraints in the instance.

Finally, in Section ~\ref{sec:dp}, we describe and analyze our DP algorithm which takes a simplified instance as input and finds a solution $R_2$ of this instance such that $R_1\cup R_2$ is a solution of the original instance. 
We also show that our overall algorithm for $(\alpha, \beta)$-{UAQ} is {\sf FPT} parameterized by $ k_r + \hat k.$ 

\subsection{Preprocessing}\label{sec:red}

In this section, we define several reduction and branching rules.
A reduction rule for an instance $(I,k)$ of a parameterized problem is {\em safe} if it reduces $(I, k)$ to $(I',k')$ such that $(I, k)$ is a yes-instance if and only if $(I', k')$ is.
A branching rule is {\em safe} if for input $(I, k)$ it outputs a number $q$ of instances $(I_1',k_1'),\dots , (I_q',k_q')$, $q \geq 1$, such that $(I,k)$ is a yes-instance if and only if $(I_i',k_i')$ is for some $i \in [q]$.
We prove that each rule is safe and can be implemented either in polynomial time (for the reduction rules) or in {\sf FPT} time (for the branching rule).
Moreover, we construct a partial solution $R_1$ of the problem, comprising those roles that must be in any solution of the problem; $R_1$ is initialized as the empty set.


{
Assume that we have an initial instance ${{\cI}}=(R,P,RP,P_{lb},\cD,k_r,k_p)$ of $(\alpha, \beta)$-{UAQ}. 
Note that the initial instance of the problem (already simplified by Reduction Rule~\ref{rule:initial-assumption})
will not be reducible by Reduction Rule~\ref{rule:initial-assumption}. 
However, a variant of that rule is still required since it may be applied to instances generated by other reduction rules.
The following Reduction Rule~\ref{rule:no-Plb-permission} is the variant of Reduction Rule~\ref{rule:initial-assumption} used in this section.

Reduction and branching rules are applied to an initial instance of the problem in the order the rules are described in this section. For every rule we take an input instance and output an instance or a number of instances (for a branching rule) such that if the output instance or one of the output instances is different from the input instance then all the previous rules are applied to the new instance (or, each of the new instances). 

\begin{reduction rule}
\label{rule:no-Plb-permission}
 (i) delete any role $r$ such that $P(r) \cap P_{lb} = \emptyset $ (since $r$ cannot contribute to a set of roles that satisfies the instance);\\
(iii) delete any role which is in a SoD constraint $\ev{X, 1}$;\\
(iv) for any deleted role $r$  and any  SoD constraint $\langle X,t\rangle$ with $r\in X$, replace the constraint by $\langle X\setminus \{r\},t \rangle$;\\
(v) for any deleted role $r$, remove all pairs of the form $(r,p)$ from $RP$;\\
(vi) delete any constraint $\ev{X, t}$ in which $|X| < t$.	
\end{reduction rule} 
}

{
We use the following {\em update procedure} for some of our subsequent rules.

\begin{reduction procedure}{\sf UPDATE}
\label{update-procedure:deleting-essential-roles}
Let $r \in R$ be an arbitrary role.
Then, Procedure~\ref{update-procedure:deleting-essential-roles}$(r)$ executes the following steps.
\begin{itemize}
	\item Add $r$ to $R_1;$
	\item Set $R \leftarrow R\setminus  \{r\}$, $k_r \leftarrow k_r - 1$, $k_p \leftarrow k_p - |P(r) \setminus P_{\ell b}|$, $P_{\ell b} \leftarrow P_{\ell b} \setminus P(r)$,  $P \leftarrow P \setminus P(r)$; 
	\item  Remove from $RP$ all pairs which include permissions from $P(r)$; and
	\item For every constraint $\ev{X, t} \in \cD$, if $r \in X$, then replace the constraint $\ev{X, t}$ by $\ev{X \setminus \{r\}, t - 1}$.
\end{itemize}

\end{reduction procedure}
}




\begin{reduction rule}
\label{rule:pendant-vertex-domination}
Suppose that there is a permission $p \in P_{\ell b}$ that is assigned to a unique role $r \in R$. 
If $|P(r) \setminus P_{\ell b}| \geq k_p +1$, then $(R,P,RP,P_{lb},\cD,k_r,k_p)$ is a no-instance.
{
Otherwise (i.e. when $|P(r) \setminus P_{\ell b}| \leq k_p$), perform Procedure~\ref{update-procedure:deleting-essential-roles}$(r)$.
}
\end{reduction rule}

\begin{lemma}
\label{lem:safety-pendant-vertex-domination-rule}
Reduction Rule~\ref{rule:pendant-vertex-domination} is safe and can be implemented in polynomial time.
\end{lemma}

\begin{proof}
For convenience we will denote the reduced instance by $(R', P', RP', P_{\ell b}', \cD', k_r', k_p')$.
When $r \in R$ is the unique role to which $p \in P_{\ell b}$ is assigned, then $r$ must be in any solution.
But, if $r$ is authorized for more than $k_p$ permissions from $P \setminus P_{\ell b}$, then $r$ cannot belong to any solution.
Thus, the input instance is a no-instance.
So, we assume that $|P(r) \setminus P_{\ell b}| \leq k_p$.

The backward direction ($\Leftarrow$) of the reduction rule is trivial.
Let $R^* \subseteq R'$ be a set of at most $k_r'$ roles such that $P(R^*) \supseteq P_{\ell b}'$ and $|P(R^*) \setminus P_{\ell b}'| \leq k'_p$.
{ Also, $R^*$ satisfies all SoD constraints in $\cD'$.}
Then, $R^* \cup \{r\}$ is a set of at most $k_r' +1 = k_r$ roles such that $P(R^* \cup \{r\}) \supseteq P_{\ell b}$ and $|P(R^* \cup \{r\}) \setminus P_{\ell b}| \leqslant k_p' + |P(r) \cap (P \setminus P_{\ell b})| = k_p$.
{ Also, if a constraint $\ev{X, t} \in \cD$ contains $r$, then this constraint $\ev{X, t}$ will have been replaced by $\ev{X \setminus \{r\}, t -1} \in \cD'$.
Since, $|(X \setminus \{r\}) \cap R^*| \leq t - 1$, we have that, $|X \cap (R^* \cup \{r\})| \leq t$.
Hence, all constraints in $\cD$ are also satisfied.
}

Now, we give the forward direction $(\Rightarrow)$ of the proof.
Let $(R, P, PR, P_{\ell b}, \cD, k_r, k_p)$ be a yes-instance.
Since $p \in P_{\ell b}$ is assigned to a unique role $r \in R$, $r$ must be in any solution of $(R, P, PR, P_{\ell b}, \cD, k_r, k_p)$.
If $R^*$ is a set of at most $k_r$ roles that is a solution for $(R, P, PR, P_{\ell b}, \cD, k_r, k_p)$, then $R^* \setminus \{r\}$ is a set of at most $k_r'$ roles such that $P(R^* \setminus \{r\}) \supseteq P'_{\ell b}$ and $|P(R^* \setminus \{r\}) \setminus P_{\ell b}'| \leq k'_p$.
{
Let $\ev{X, t} \in \cD$ be an SoD constraint such that $r \in X$.
Then, $\ev{X, t}$ will have been replaced by $\ev{X \setminus \{r\}, t - 1} \in \cD'$.
Hence, $\ev{X \setminus \{r\}, t - 1}$ is satisfied by $R^* \setminus \{r\}$ and all SoD constraints in $\cD'$ are also satisfied.
}
\end{proof}

\begin{branching rule}
\label{branch-rule:for-i-sized-sets} Let $b = \beta k_r^q+ \sum\limits_{a = 1}^{q - 1} k_r^a$.
For $q = 1, 2, \ldots, \alpha -2$, in this order, apply Branching Rule $1.q$ repeatedly until it no longer causes any changes to the RPG.
\begin{itemize}
	\item {\bf Branching Rule $\mathbf{1.q}$:} 
	If there exists a set $L \subseteq R$ with $(\alpha - q)$ roles such that $|(\bigcap\limits_{r \in L} P(r)) \cap P_{\ell b}| > b$, then 
		apply the following branching: for every role $r$ in $L$ such that $|P(r) \setminus P_{\ell b}| \leq k_p$,		perform Procedure~\ref{update-procedure:deleting-essential-roles}$(r)$.
\end{itemize}
\end{branching rule}

\begin{lemma}
\label{claim:i-sized-set-intersection}
Let $R^* \subseteq R$ be a set of at most $k_r$ roles that satisfies all constraints in $\cD$, where $P(R^*) \supseteq P_{\ell b}$ and $|P(R^*)| \leq |P_{\ell b}| + k_p$.
Consider an application of Branching Rule~\ref{branch-rule:for-i-sized-sets}.$q$, $1 \leq q \leq \alpha - 2$.
If $L$ is a set of roles from $R$ that satisfies the condition in Branching Rule $1.q$, then $ R^*\cap L \neq \emptyset$.
Furthermore, $R^* \cap L$ is authorized for at most $k_p$ permissions from $P \setminus P_{\ell b}$.
\end{lemma}

\begin{proof}
Let $\hat P_{\ell b} = (\bigcap_{r \in L} P(r)) \cap P_{\ell b}$. 
Let $R^* \subseteq R$ be a set of at most $k_r$ roles that satisfies all constraints in $\cD'$, is authorized for all permissions in $P_{\ell b}$, and at most $k_p$ permissions from $P \setminus P_{\ell b}$.
Consider two cases.

\noindent{\bf Case 1:} $q = 1$.
Then $|L| = \alpha - 1$.
Suppose that $R^* \cap L = \emptyset$.
Consider any role $r \in R^*$.
Since $RPG$ is $K_{\alpha,\beta}$-free, $|P(r) \cap \hat P_{\ell b}| \leq \beta - 1$.
Then, $R^*$ can be authorized for at most $(\beta-1)k_r$ permissions from $\hat P_{\ell b}$.
This implies that there exists some permission $p \in \hat P_{\ell b} \subseteq P_{\ell b}$ such that $R^*$ cannot authorize $p$.
This is a contradiction.
Hence, $L \cap R^* \neq \emptyset$.

\vspace{2mm}

\noindent{\bf Case 2:}  $2 \leq q \leq \alpha - 2$.
Suppose that $R^* \cap L = \emptyset$.
Since $|\hat P_{\ell b}| > b$, by the pigeon hole principle, there exists a role $y \in R^*$ such that $y$ is authorized for at least $b/k_r + 1 = \beta k_r^{q-1} + k_r^{q-2} + \cdots + k_r + 1 + 1$ permissions in $\hat P_{\ell b}$.
Then, consider $P(L \cup \{y\})$.
There are $(\alpha-q+1)$ roles in $L \cup \{y\}$ and $L \cup \{y\}$ is authorized for at least $\beta k_r^{q-1} + k_r^{q-2} + \cdots + k_r + 1 + 1$ common permissions in $P_{\ell b}$.
Then, Branching Rule~\ref{branch-rule:for-i-sized-sets}.$(q-1)$ is also applicable.
But, we apply Branching Rule~\ref{branch-rule:for-i-sized-sets}.$q$ only when for every $i \in [q-1]$, Branching Rule~\ref{branch-rule:for-i-sized-sets}.$i$ is not applicable.
This is a contradiction.
Hence, $R^*$ has nonempty intersection with $L$.

As $R^*$ is authorized for at most $k_p$ permissions from $P \setminus P_{\ell b}$, $L \cap R^*$ is authorized for at most $k_p$ permissions in $P \setminus P_{\ell b}$.
This completes the proof.
\end{proof}

\begin{reduction rule}
\label{rule:high-degree-removal}
Suppose that there is a role $s \in R$ such that $s$ is authorized for at least $h = \beta k_r^{\alpha-1} + k_r^{\alpha - 2} + \cdots + k_r^2 + k_r + 1$ permissions in $P_{\ell b}$.
If $s$ is authorized for more than $k_p$ permissions from $P \setminus P_{\ell b}$, then $(R, P, RP,  P_{\ell b}, \cD, k_r, k_p)$ is a no-instance.
{ Otherwise, perform Procedure~\ref{update-procedure:deleting-essential-roles}$(s)$.}
\end{reduction rule}

\begin{lemma}
\label{lemma:safeness-high-degree-removal-rule}
Reduction Rule~\ref{rule:high-degree-removal} is safe and can be implemented in polynomial time.
\end{lemma}
\begin{proof}
It suffices to prove that any solution of $(R, P, RP, P_{\ell b}, \cD, k_r, k_p)$ must contain $s$.
Let $A = P(s).$
Suppose that there exists a solution $S \subseteq R$ such that $s \notin S$ and $|S| \leq k_r$.
Recall that $RPG$ is $K_{\alpha, \beta}$-free.
By assumption, $S$ is authorized for all permissions in $A$.
Since $|A| \geq h$ (by the precondition of the reduction rule), there exists $s' \in S$ such that $s'$ is auhorized for at least $\ceil{h/k_r} = \beta k_r^{\alpha-2} + k_r^{\alpha-3} + \cdots + k_r^2 + 1 + 1$ permissions from $P_{\ell b}$.
Then, $R$ has $r_1, r_2 \in S$ such that $P(r_1) \cap P(r_2)$ contains at least $\beta k_r^{\alpha-2} + k_r^{\alpha-3} + \cdots + k_r^2 + 1 + 1$ permissions.
This means that the precondition of Branching Rule~\ref{branch-rule:for-i-sized-sets}.$q$ becomes applicable for $q = \alpha -2$.
This is a contradiction.
Hence, $s \in S$.
\end{proof}

{
After Reduction Rule~\ref{rule:pendant-vertex-domination}, Branching Rule~\ref{branch-rule:for-i-sized-sets}, and Reduction Rule~\ref{rule:high-degree-removal} have been performed, there is still a possibility 
that there is a role $r$ such that $|P(r) \setminus P_{\ell b}| > k_p$.
However, such roles cannot be included in any solution.
Hence, we need to apply the following reduction rule.
}

\begin{reduction rule}
\label{rule:removal-of-forbidden-roles}
Delete  any role $r$ such that $|P(r) \setminus P_{lb}| > k_p$ (since $r$ cannot be included in any yes-instance).
Remove any pair from $RP$ that contains a deleted role.
If there is a deleted role $r \in X$ for a constraint $\ev{X, t}$, then replace this constraint by $\ev{X \setminus \{r\}, t}$. 
\end{reduction rule}


%
%
%

{
\begin{lemma}
\label{lemma:reduced-instance-size-kij-free}
Suppose that $\cI = (R, P, RP, P_{\ell b},  \cD, k_r, k_p)$ is an $(\alpha,\beta)$-{\UAQ} instance for which Reduction Rules~\ref{rule:no-Plb-permission},~\ref{rule:pendant-vertex-domination},~\ref{rule:high-degree-removal}, and~\ref{rule:removal-of-forbidden-roles}, and Branching Rule~\ref{branch-rule:for-i-sized-sets} are not applicable.
If $\left|P_{\ell b}\right| > \beta k_r^{\alpha} + k_r^{\alpha-1} +\cdots + k_r^2 + k_r$, then ${{\cI}}$ is a no-instance.
\end{lemma}

\begin{proof}
Let $\cI = (R, P, RP, P_{\ell b},  \cD, k_r, k_p)$ be an irreducible yes instance but $\left|P_{\ell b}\right| > \beta k_r^{\alpha} + k_r^{\alpha-1} +\cdots + k_r^2 + k_r$.
Suppose that  $S \subseteq R$ is an arbitrary set such that $|S| \leq k_r$ and it satisfies all the SoD constraints.
Since Reduction Rule~\ref{rule:high-degree-removal} is not applicable to ${\cI}$, $$\bigg|\bigcup\limits_{s \in S} P(s) \cap P_{\ell b}\bigg| \leq \beta k_r^{\alpha} + k_r^{\alpha} + k_r^{\alpha-1} +\cdots + k_r^2 + k_r.$$

Then, $S$ is not authorized for all permissions in $P_{\ell b}.$
This contradicts the fact that $S$ is a feasible solution to $\cI$.
Therefore, $\cI$ is a no-instance.
\end{proof}
}

Observe that we have not yet used any special characteristics of the set $\cD$ of SoD constraints.
Hence, Lemma~\ref{lemma:reduced-instance-size-kij-free} holds true for all classes of SoD constraints.
Thus we obtain the following result.

\begin{theorem}
\label{thm:partial-uaq-implication}
The  preprocessing described above runs in time $\cO^*(\alpha^{k_r})$ and for an input instance of $(\alpha,\beta)$-{UAQ}, outputs an equivalent instance of $(\alpha,\beta)$-{UAQ} such that $|P_{\ell b}|=\cO(\beta k_r^{\alpha})$.
\end{theorem}

\begin{proof}
Observe that all the reduction rules above can be implemented in polynomial time.
For every $q \in [\alpha-2]$, Branching Rule~\ref{branch-rule:for-i-sized-sets}.$q$ is applied only when for all $x \in [q-1]$, Branching Rule~\ref{branch-rule:for-i-sized-sets}.$x$ is not applicable.
This branching rule is applied on a set $L$ with at most $\alpha - 1$ vertices.
From Lemma~\ref{claim:i-sized-set-intersection}, if the precondition to Branching Rule~\ref{branch-rule:for-i-sized-sets} is satisfied, then $L$ has nonempty intersection with a solution of size at most $k_r$.
Hence, this branching rule provides a $(\alpha-1)$-way branching with depth at most $k_r$.
Hence, the number of leaves in this bounded search tree is $\cO^*(\alpha^{k_r})$.

{
We first apply Reduction Rules~\ref{rule:no-Plb-permission},~\ref{rule:pendant-vertex-domination}, Branching Rule~\ref{branch-rule:for-i-sized-sets}, and Reduction Rules~\ref{rule:high-degree-removal}, and~\ref{rule:removal-of-forbidden-roles} in sequence.
When these reduction and branching rules are no longer applicable, we check whether $|P_{\ell b}| > \beta k_r^{\alpha} + k_r^{\alpha} + k_r^{\alpha-1} +\cdots + k_r^2 + k_r$. 
If $|P_{\ell b}| > \beta k_r^{\alpha} + k_r^{\alpha} + k_r^{\alpha-1} +\cdots + k_r^2 + k_r$, then we use Lemma~\ref{lemma:reduced-instance-size-kij-free} to output that the given input instance is a no-instance.
Otherwise, $|P_{\ell b}| \leq \beta k_r^{\alpha} + k_r^{\alpha} + k_r^{\alpha-1} +\cdots + k_r^2 + k_r$.}

Hence, the reduced instance has $|P_{\ell b}|=\cO(\beta k_r^{\alpha}).$
Since $\alpha$ is fixed, the input instance can be transformed into an equivalent instance in $\cO^*(\alpha^{k_r})$ time.
This completes the proof.
\end{proof}

\subsection{Encoding Candidate Solutions as a Matroid}\label{sec:pdp}

This section consists of two parts: in Section~\ref{subsec:matroid-dp}, we provide terminology, notation and results on matroids and representative families, which are necessary for describing and analyzing our DP algorithm;
and in Section~\ref{subsec:partition-matroid}, we construct a matroid encoding role sets that do not violate SoD constraints.
{
We will call this matroid a {\em constraint satisfaction matroid (CSM)}.
}
This constraint satisfaction matroid is used in our DP algorithm.

\subsubsection{Matroids and Representative Families}\label{subsec:matroid-dp}

\begin{definition}\label{defn:matroid}
A family of sets $\cI$ over a finite universe $U$ is called a {\em matroid} if it satisfies the following three axioms:
\begin{enumerate}
	\item\label{matroid-prop:basic} $\emptyset \in \cI$,
	\item\label{matoid-prop:hereditary} if $A \in \cI$ and $B \subseteq A$, then $B \in \cI$, and
	\item\label{matoid-prop-exchange} if $A, B \in \cI$ such that $|A| < |B|$, then there exists $x \in B \setminus A$ such that $A \cup \{x\} \in \cI$.
\end{enumerate}
\end{definition}

{
Let $U$ be a universe of $n$ elements, and let $r$ be an integer such that $r\le n$. 
Then it is not hard to verify that $(U, \cF)$ where $\cF = \{A \subseteq U:\ |A| \leq r\}$ satisfies the axioms of Definition~\ref{defn:matroid}.
Hence, $(U, \cF)$ is a matroid; $(U, \cF)$ is called a {\em uniform matroid}.}

For a matroid $M = (U, \cI),$ any set $A \in \cI$ is called an {\em independent set}.
It follows from Definition~\ref{defn:matroid} that all maximal independent sets of a matroid $M$ have the same size, denoted $rank(M)$, and called the {\em rank} of $M$. 
Clearly, the rank of the uniform matroid $(U, \cF) = \{A \subseteq U:\ |A| \leq r\}$ is $r$.
  
A matroid $M=(U, \cI)$ is said to be {\em representable over a field $\bbF$} if there is a matrix $\hat M$ over $\bbF$ and a bijection $\pi:\ U\to {\rm col}(\hat{M})$, where ${\rm col}(\hat{M})$ is the set of columns of $\hat M$, such that $A \subseteq U$ is an independent set in $M$ if and only if $\{\pi(a):\ a\in A\}$ is linearly independent over $\bbF$. 
Clearly, the rank of $M$ is the rank of the matrix $\hat M$.
A matroid representable over a field $\bbF$ is a {\em linear matroid} over $\bbF$. 
A uniform matroid with $U$ of size $n$ can be represented over any field ${GF}(p$) for $p>n$ (see e.g.  \cite{CyganFKLMPPS15}).

\begin{definition}
\label{defn:direct-sum-of-matroid}
Let $M_1 = (U_1, \cI_1), M_2 = (U_2, \cI_2),\ldots,M_t = (U_t, \cI_t)$ be a collection of matroids such that for every $i \neq j$, $U_i \cap U_j = \emptyset$.
Then, the {\em direct sum} $M = M_1 \oplus M_2 \oplus \cdots \oplus M_t$ of these matroids is a matroid $M = (U, \cI)$ such that $U = U_1 \cup U_2 \cup \cdots \cup U_t$, and for every subset $S$ of $U$, $S \in \cI$ if and only if for all $i \in [t]$, $S \cap U_i \in \cI_i$.

A {\em partition matroid} is a matroid formed from a direct sum of uniform matroids.
\end{definition}

%

We use the following definitions and results~\cite{LokshtanovMPS18,FominLPS16} to prove Theorem~\ref{thm:search-version-solving} in Section~\ref{sec:dp}.

\begin{definition}
\label{defn:q-rep-family}
Let $M = (U, \cI)$ be a matroid and $\AAA$ be a family of independent sets of size $p$ in $M$.
For sets $A, B \subseteq U$, we say that {\em $A$ fits $B$} if $A \cap B = \emptyset$, and $A \cup B \in \cI$.

{ A subfamily $\hat \AAA \subseteq \AAA$ is a {\em $q$-representative family of $\AAA$} if the following holds:
for every $B \subseteq U$ with $|B| \le q$, there is an $A \in \AAA$ such that $A$ fits $B$ if and only if there is an $\hat A \in \hat \AAA$ such that $\hat A$ fits  $B$.
We write $\hat \AAA \subseteq_{rep}^q \AAA$ to denote that $\hat \AAA$ is a $q$-representative family of $\AAA$.}
\end{definition}


Informally, a family of sets $\AAA$ that fits $\{B \subseteq U : |B| = q\}$ provides a way of encoding all sets of cardinality $q$ in a matroid.
Thus, a $q$-representative family is a compact method of encoding all such sets.
Our dynamic programming algorithm makes use of $q$-representative families and the following results~\cite{FominLPS16,LokshtanovMPS18} to reduce the number of entries in the DP table and ensure our algorithm is {\sf FPT}. 

\begin{lemma}
\label{lemma:q-rep-property}
Let $M = (U, \cI)$ be a matroid and $\cF \subseteq \cI$ such that for every $A \in \cF$, $|A| = p$.
If $\cF_1 \subseteq_{rep}^q \cF$ and $\cF_2 \subseteq_{rep}^q \cF_1$, then $\cF_2 \subseteq_{rep}^q \cF$.
\end{lemma}

\begin{lemma}
\label{lemma:q-rep-family-computation}
Let $M = (U, \cI)$ be a linear matroid of rank $n$. Suppose that $M$ can be represented by an $n \times |U|$-matrix $\hat M$ over a field $\mathbb{F}$ such that $\mathbb{F} = GF(s)$ or $\mathbb{F} = \mathbb{Q}$.
Furthermore, let $\cF = \{S_1,\ldots,S_t\}$ be a family of independent sets in $M$, each of cardinality $p$.
Then, there is a deterministic algorithm that computes $\hat \cF \subseteq_{rep}^q \cF$ with $\cO({{p+q}\choose{p}}^2 tp^3 n^2 + t{{p+q}\choose{p}}^{\omega} np) + (n + |U|)^{\cO(1)}$ field operations over $\mathbb{F}$ such that $|\hat \cF| \leq {{p+q}\choose{p}}$ and $\omega < 2.37$ is the matrix multiplication exponent.
\end{lemma}

\subsubsection{Constraint Satisfaction Matroid for $(\alpha, \beta)$-{UAQ}}\label{subsec:partition-matroid}

Let $\cD' = \{\ev{X_1, t_1},\ldots,\ev{X_m, t_m}\}$ be the collection of all SoD constraints in the reduced instance. Recall that for all $i, j \in [m]$ with $i \neq j$, $X_i \cap X_j = \emptyset$.
For every $i \in [m]$, define $M_i(\cD') = (X_i, \cI_i)$, where $\cI_i = \{A \subseteq X_i :\ |A| \leq t_i - 1\}$.
Observe that $M_i(\cD')$ is a uniform matroid of rank $t_i - 1$ and each element of the matroid is a set of roles (in $X_i$) to which a single user could be assigned. 

Let $\{r_{m+1},r_{m+2}, \dots , r_{\delta}\}$ be the set of roles that do not appear in any $X_i$.
(That is, $\{\{r_{m+1},\dots,r_{\delta}\} \cup X_1 \cup \dots \cup X_m = R$.)
For $r_j$, $m+1 \leq j \leq \delta$, we create a constraint $\ev{X_j, t_j}$ such that $X_j = \{r_j\}$ and $t_j = 2$, and construct the uniform matroid $M_j(\cD') = (X_j,\{\emptyset, \{r_j\}\})$. 
	
Now let $M(\cD')= M_1(\cD') \oplus \cdots \oplus M_{\delta}(\cD')$.
By construction, $M(\cD') = (R, \cI)$ is a partition matroid, where $\cI =  \{B \subseteq R: |B \cap X_i| \leq t_i - 1, i\in [\delta]\}$.
Notice that each set in this matroid is a set of roles that could be assigned to a single user (since the constraint sets are assumed to be disjoint).

It is known that a partition matroid $M = (U,{\cI})$ is linear and that it possible to compute a $(|U| \times |U|)$-matrix that represents a partition matroid over $GF(p)$ for any $p > |U|$ in time polynomial in $|U|$~\cite{CyganFKLMPPS15}.
%
Hence, we have the following result.

\begin{lemma}
\label{lemma:its-a-partition-matroid}
The matroid $M(\cD') = (R, \cI)$ is a linear matroid of rank at most $|R|$.
An $|R| \times |R|$-matrix $\hat M$ over $GF(p)$ for some $p > |R|$ representing $M(\cD')$ can be constructed in time polynomial in $|R|$.
\end{lemma}

{ 
We call $M(\cD')$ the constraint satisfaction matroid (CSM) for $\cD'$.
Observe that $M(\cD')$ is represented by an $(|R| \times |R|)$-matrix.
This matrix provides a compact representation of all members in $\cI$:
a subset of roles $R' \in \cI$ if and only if the columns in $M(\cD')$ representing the elements of $R'$ are linearly independent.
Hence, given an arbitrary $R' \subseteq R$, there exists an algorithm that runs in $\cO(|R|^{\cO(1)})$ time and correctly outputs whether $R' \in \cI$ or not (see~\cite{CyganFKLMPPS15}).} 

\subsection{Algorithm Description and Analysis}\label{sec:dp}

The {\em whole algorithm} for $(\alpha,\beta)$-UAQ starts from the preprocessing described in Section~\ref{sec:red} and returns a partial solution $R_1$ or concludes that the input instance  is a no-instance.
If $R_1$ is returned, then the algorithm constructs a constraint satisfaction matroid as in Section~\ref{sec:pdp}, and then uses the DP algorithm described below to produce a partial solution $R_2$ or concludes that the input instance is a no-instance. 
If $R_2$ is returned, then the whole algorithm returns $R_1\cup R_2$.
We prove correctness of the whole algorithm and evaluate its running time in Theorem~\ref{thm:search-version-solving}.

Our algorithm considers possible sets of extra permissions that could be included in a solution.
Such a set must be a subset of $P \setminus P_{\ell b}$ and have cardinality no greater than $k_p$.
Accordingly, we define $P_{\it good} = \{Y \subseteq P \setminus P_{\ell b} : |Y| \leq k_p\}$.

Let $W \subseteq P_{\ell b}, Y \in P_{\it good}$ and $0 \leq i  \leq k_r$. 
We define the following set:
  \[
   \mathcal{B}[W,Y,i] = \{R' \subseteq R \mid W \subseteq P(R') \subseteq P_{\ell b} \cup Y, |R'| = i, R' \in \mathcal{I}\},
  \] 
where $\mathcal{I}$ is the family of independent sets of $M(\cD')$.
 
 Note that if there exist $i \in \{0,\dots,k_r\}$ and  $Y \in P_{\it good}$ such that $\mathcal{B}[P_{\ell b},Y, i]$ is non-empty, then any member of $\BB[P_{\ell b}, Y, i]$ provides a solution to the reduced (by the preprocessing) problem.
 Let $\hat \BB[P_{\ell b}, Y, i]  \subseteq_{rep}^{k_r - i} \BB[P_{\ell b}, Y, i]$.
 Observe that by Definition~\ref{defn:q-rep-family}, $\hat \BB[P_{\ell b}, Y, i] \neq \emptyset$ if and only if $\BB[P_{\ell b}, Y, i] \neq \emptyset$.
  
The following observation is not hard to verify.

\begin{proposition}
\label{obs:when-i-is-zero}
Let $Y \in P_{\it good}$ and $W \subseteq P_{\ell b}$. Then, the following statements hold true for $i = 0$:
\begin{itemize}
    \item if $W = \emptyset$, 
    then $\hat{\mathcal{B}}[W,Y,0] = \mathcal{B}[W,Y,0] = \{\emptyset\}$;
    \item if $W \ne \emptyset$, 
    then $\hat{\mathcal{B}}[W,Y,0] = \mathcal{B}[W,Y,0]  = \emptyset$.
\end{itemize}
\end{proposition}
\begin{proof}
	If $W = \emptyset$, then $\mathcal{B}[W,Y,0] = \{\emptyset\}$ (since no role is required to authorize an empty set of permissions).
	Hence, from Definition~\ref{defn:q-rep-family}, we obtain that $\hat \BB[W, Y, 0] = \{\emptyset\}$.
	 If $W \neq \emptyset$, then at least one role is required to authorize the permission set $W.$
	Therefore, $\BB[W, Y, 0] = \emptyset$.
	Hence, by Definition~\ref{defn:q-rep-family}, we obtain that $\hat \BB[W, Y, 0] = \emptyset$.
\end{proof}

We now describe our DP algorithm.

\begin{enumerate}
	\item First, for all $\emptyset \ne W\subseteq P_{\ell b}$ and $Y \in P_{\it good}$, we initialize $\hat{\mathcal{B}}[\emptyset, Y, 0] = \{\emptyset\}$ and $\hat{\mathcal{B}}[W, Y, 0] = \emptyset$.
 	\item  Then, for all $i = 1, 2,\ldots, k_r$, for every $Y \in P_{\it good}$ and for every $W \subseteq P_{\ell b}$, we compute $\hat \BB[W, Y, i]$ as follows:
	\begin{enumerate}
        \item if for all $r \in R$ such that $P(r) \subseteq Y \cup P_{\ell b}$, we have $\hat{\mathcal{B}}[W \setminus P(r),Y,i-1] = \emptyset$, then we set $\hat{\mathcal{B}}[W,Y,i] = \emptyset$;
  
        \item otherwise there exists $r \in R$ such that $P(r) \subseteq P_{\ell b} \cup Y$ and $\hat \BB[W \setminus P(r), Y, i-1] \neq \emptyset$.
            Then, we compute  $\mathcal{X}[W, Y, i]$ as
            \begin{equation}\label{equation-1}
                \{A \cup \{r\} \mid P(r) \subseteq P_{\ell b} \cup Y,   \hat{\BB}[W \setminus P(r), Y, i-1] \neq \emptyset, A \in \hat{\mathcal{B}}[W \setminus P(r), Y, i-1], r \in R \setminus A\} \cap \cI
            \end{equation}
            and $\hat \BB[W,Y,i]$ as
            $\hat{\mathcal{B}}[W,Y,i] \subseteq_{\it rep}^{k_r - i} \mathcal{X}[W,Y,i]$. 
            
            If $\hat{\mathcal{B}}[P_{\ell b},Y,i]\neq \emptyset$ for some  $Y \in P_{\it good}$ then return any role set in $\hat{\mathcal{B}}[P_{\ell b},Y,i]$ as $R_2$ and halt. 
    \end{enumerate}
 \item Return ``no-instance''. 
\end{enumerate}

  Note that $\cX[W,Y,i]$ could contain up to $|R|\binom{k_r}{i-1}$ sets.
  But by Lemma~\ref{lemma:q-rep-family-computation}, the size of $\hat{\mathcal{B}}[W,Y,i]$ is at most just $\binom{k_r}{i}$.
  So, we invoke Lemma~\ref{lemma:q-rep-family-computation}, and store $\hat{\mathcal{B}}[W,Y,i]$ such that $\hat \BB[W, Y, i] \subseteq_{rep}^{k_r - i} \cX[W, Y, i]$ (informally, $\hat \BB[W, Y, i]$ is a $(k_r - i)$-representative family of $\mathcal{X}[W,Y,i]$).
 
\begin{lemma}
\label{lemma:rep-set-equivalence}
For every $W \subseteq P_{\ell b}$, $Y \in P_{\it good}$  and $0 \leq i \leq k_r$, we have $\hat \BB[W, Y, i]\subseteq_{rep}^{k_r - i} \mathcal{B}[W, Y, i].$
\end{lemma}

\begin{proof}
Let $W \subseteq P_{\ell b}$, and $Y \in P_{\it good}$.
	We prove this lemma by induction on $i$.
	
	\vspace{1mm}
	
	{\bf Base Case:} The case $i = 0$ holds true due to Proposition~\ref{obs:when-i-is-zero}.
	
	\vspace{1mm}
	
	{\bf Induction Hypothesis:} Let $i \geq 1$, and assume by induction hypothesis that for all $j < i$, and for all $W' \subseteq W$, $\hat \BB[W', Y, j] \subseteq_{rep}^{k_r - j} \BB[W', Y, j]$.
	
	\vspace{1mm}
	
	{\bf Induction Step:} Let $i \geq 1$. Recall that by the DP algorithm  description $\hat \BB[W, Y, i] \subseteq_{rep}^{k_r-i} \cX[W, Y, i]$.
	Thus, by Lemma~\ref{lemma:q-rep-property},
	if  $\cX[W, Y, i] \subseteq_{rep}^{k_r-i} \BB[W, Y, i]$ holds true then $\hat \BB[W, Y, i] \subseteq_{rep}^{k_r-i} \BB[W, Y, i].$ Therefore, in the rest of the proof it suffices to show that $\cX[W, Y, i] \subseteq_{rep}^{k_r-i} \BB[W, Y, i].$
	
	Let $X_1 = \{r_1,\ldots,r_i\} \in \BB[W, Y, i]$ and $X_1' = X_1 \setminus \{r_i\}$.
	Let $X_2$ be a set of at most $k_r-i$ roles such that $X_1 \cap X_2 = \emptyset$, and $X_1 \cup X_2 \in \cI$ and 
	let $X_2^* = X_2 \cup \{r_i\}$.
	Since $r_i$ authorizes $P(r_i)$, and $X_1'$ authorizes all permissions in $W \setminus P(r_i)$ and $P(X_1') \setminus P_{\ell b} \subseteq Y$, we have that $X_1' \in \BB[W \setminus P(r_i), Y, i-1]$.
	
	By induction hypothesis, $\hat \BB[W \setminus P(r_i), Y, i-1] \subseteq_{rep}^{k_r - i + 1} \BB[W \setminus P(r_i), Y, i-1]$.
	Hence, by Definition~\ref{defn:q-rep-family}, there exists $X_1^* \in \hat \BB[W \setminus P(r_i), Y, i-1]$ such that $X_1^* \cup X_2^* \in \cI$ and $X_1^* \cap X_2^* = \emptyset$.
	By (\ref{equation-1}), we have $X_1^* \cup \{r_i\} \in \cX[W, Y, i]$.
	The set of roles $X_1^* \cup \{r_i\}$ is such that $(X_1^* \cup \{r_i\}) \cap X_2 = \emptyset$ and $X_1^* \cup \{r_i\} \cup X_2 \in \cI$.
	Hence, $\cX[W, Y, i] \subseteq_{rep}^{k_r-i} \BB[W, Y, i]$. 
	This completes the proof.
\end{proof}

\begin{theorem}
\label{thm:search-version-solving}
The whole algorithm solves $(\alpha,\beta)$-{UAQ} in time
 $\mathcal{O}^*(2^{\mathcal{O}(k_r^{\alpha} + \hat k)})$.
\end{theorem}

\begin{proof}
 By Lemma~\ref{lemma:rep-set-equivalence}, we have $\hat{\mathcal{B}}[W,Y,i] \subseteq_{\it rep}^{k_r - i} \mathcal{B}[W,Y,i]$.
 Hence, the algorithm correctly computes partial solutions and if there exists $R_2 \in \hat{\mathcal{B}}[P_{\ell b},Y,i]$ for some $Y \in P_{\it good}$ and for some $0 \leq i \leq k_r$, then $R_1 \cup R_2$ is a solution to the UAQ instance. (If no such $R_2$ exists for any appropriate $Y$ then there is no solution to the instance.)
 Thus, the algorithm is correct.
 
 The time taken to compute a solution is determined by the running times of the pre-processing phase ($\mathcal{O}(\alpha^{k_r})$), the computation of the matroid (polynomial in $|R|$) and the dynamic programming phase.
 
 From Theorem~\ref{thm:partial-uaq-implication}, we have $|P_{\ell b}|$ is $\cO(k_r^{\alpha})$.
 Hence, there are $2^{|P_{\ell b}|} = 2^{\mathcal{O}(k_r^{\alpha})}$ subsets of $P_{\ell b}$ and there are at most $2^{\hat{k}}$ subsets in $P \setminus P_{\ell b}$.
 Hence, there are $\mathcal{O}^*(2^{\mathcal{O}(k_r^{\alpha} + \hat{k})})$ sets of the form $\mathcal{B}[W,Y,i]$.
 By~\eqref{equation-1}, 
 \[
  |\mathcal{X}[W,Y,i]| \leq |R| \max_{r \in R} |\hat{\mathcal{B}}[W \setminus P(r),Y,i-1]| = \mathcal{O}(2^{k_r} |R|).
 \]
Computing $\mathcal{X}[W, Y, i]$ takes time polynomial  in $|\mathcal{X}[W, Y, i]|$ and thus time $\mathcal{O}^*(2^{\mathcal{O}(k_r)}).$

We then compute $\hat{\mathcal{B}}[W, Y, i]$ such that $\hat{\mathcal{B}}[W, Y, i] \subseteq_{rep}^{k_r - i} \mathcal{X}[W, Y, i]$ and store it in $\hat{\mathcal{B}}[W, Y, i]$.
Recall that from Lemma~\ref{lemma:its-a-partition-matroid}, our matroid is represented by an $|R| \times |R|$-matrix.
By Lemma~\ref{lemma:q-rep-family-computation},  computing $\hat{\mathcal{B}}[W, Y, i]$ takes time \[\mathcal{O}^*(2^{\mathcal{O}(\omega k_r)}|\mathcal{X}[W, Y, i]|(|R| + i)^{\mathcal{O}(1)})=\mathcal{O}^*(2^{\mathcal{O}(k_r)}).\] 
Hence, computing every table entry takes time \[\mathcal{O}^*(2^{\mathcal{O}(k_r^{\alpha} + \hat k)}2^{\mathcal{O}(k_r)})=\mathcal{O}^*(2^{\mathcal{O}(k_r^{\alpha}  + \hat k)}).\]
Therefore, we can complete the dynamic programming phase (and hence the whole algorithm) in $\mathcal{O}^*(2^{\mathcal{O}(k_r^{\alpha} + \hat k)})$ time.
\end{proof}

\section{Related Work}

The study of \UAQ\ began with work by Du and Joshi~\cite{DuJo06} in the context of inter-domain role mapping.
They showed that finding a minimal set of roles $R'$ such that $P(R') = P_{\ell b}$ is {\sf NP}-hard via a reduction from {\sc Minimal Set Cover} and proposed a polynomial-time algorithm for computing approximate solutions.
Crampton and Chen showed that several versions of UAQ, including ones where $P(R')$ may be a superset of $P_{lb}$ were {\sf NP}-hard.
This early work did not consider SoD constraints.

Wickramaarachchi et al~\cite{ZhangJ08,WiQaLi09} extended the definition of UAQ to include SoD constraints and developed exact algorithms to solve UAQ, based on techniques used to solve CNF-SAT and MAXSAT.
Armando {\em et al.} and Lu {\em et al.} made improvements to these algorithms~\cite{ArRaTu12,LuHaChHu12,LuJoJiLi15}.
Recent work has provided a comprehensive comparative analysis of methods for solving \UAQ\ and developed a set of benchmarks for evaluating \UAQ\ solvers~\cite{ArmandoGT20}.
These results suggest that a reduction of \UAQ\ to {\sc PMaxSat} is currently the most effective way of solving \UAQ.

The focus of the above work was finding approximate and exact algorithms to solve \UAQ.
Mousavi and Tripunitara were the first to consider the parameterized complexity of UAQ~\cite{MoTr12}, and included constraints in the specification of the problem.
They showed the problem of deciding whether an instance has a solution is {\sf FPT} if $P_{ub}$ is the small parameter, essentially by considering all subsets of $P_{ub}$.
While it may be reasonable in certain cases, in general $P_{ub}$ is not necessarily small.

In summary, existing work on \UAQ\ has mainly attempted to exploit existing algorithms for related problems in order to solve \UAQ, without attempting to understand the inherent difficulty of \UAQ.
In particular, there has been little effort to better understand the complexity of \UAQ\ in terms of each of its parameters.
The exception to this is the work of Mousavi, which does explore how the complexity of the problem is affected by the different parameters, although most of this work used traditional methods of complexity analysis~\cite{Mousavi14}.

\section{Concluding Remarks}
Our work provides the first thorough attempt to investigate \UAQ\ using multi-variate complexity analysis.
Our results suggest that it may be difficult to obtain a practical {\sf FPT} algorithm for \UAQ\ in general.
However, we have also shown that if an RBAC configuration satisfies certain properties then we may be able to use an {\sf FPT} algorithm to solve instances of \UAQ\ for that configuration.

One surprising conclusion of our work is the sharp contrast in {\sf FPT} results for the workflow satisfiability problem (WSP) and \UAQ.
Informally, given a set of tasks $T$, a set of users $U$, an authorization relation $A \subseteq U \times T$ (where $u$ is authorized to perform $t$ if and only if $(u,t) \in A$), and a set of constraints $C$, an instance of WSP asks whether there exists a mapping $\pi : U \rightarrow T$ such that all constraints in $C$ are satisfied and $(\pi(t),t) \in A$ for all $t$.
WSP constraints can, for example, require that the same user is not assigned to two particular tasks (a simple form of separation of duty), although considerably more complex constraints are possible.

It is relatively easy to show that WSP is {\sf NP}-hard, even when constraints are limited to the simple separation of duty constraints described above~\cite{WaLi10}.
Nevertheless, subsequent research has shown that WSP is {\sf FPT} (when the number of tasks is the small parameter) for all user-independent constraints~\cite{CoCrGaGuJo14}, which include the aforementioned simple separation of duty constraints as well as most other constraints that are likely to arise in practice.
Moreover, the {\sf FPT} algorithms for WSP are not just of theoretical interest. 
Experimental evaluations have shown that these algorithms provide a more efficient solution for WSP than brute force algorithms and methods using SAT solvers~\cite{CoCrGaGuJo14,CohenCGGJ16,KarapetyanPGG19}.

On the face of it, WSP appears to be more complex than \UAQ, not least because the constraints in a WSP instance may be much more varied than those appearing in a UAQ instance.
And both problems require us to compute a solution set that is constrained by a binary relation (the RPG in \UAQ\ and the authorization relation in WSP) and a set of constraints.
However, our results show that \UAQ\ remains a hard problem for many RBAC configurations.

Informally, the source of the complexity seems to arise from the consequences of choosing a particular element in a potential solution. 
In the case of WSP, choosing a user to perform a specific task only means that we have to check attempts to allocate that user to other tasks, in order to determine whether such an allocation would violate a constraint.
This means that we can compute all partitions of the set of steps such that each block in the partition could be assigned to a particular user; roughly speaking, this is the basis of the {\sf FPT} algorithms for WSP.
In contrast, selecting a role, so that a particular permission in $P_{\ell b}$ is activated, means (i)~that potentially many other permissions may be simultaneously activated, and (ii)~other roles may become ineligible for consideration because of the SoD constraints.

It would be convenient if $|R|$ were the small parameter.
We could compute a solution to \UAQ\ simply by considering all possible subsets of $R$.
But the application of Reduction Rule 0 only eliminates roles that are assigned to a permission outside $P_{ub}$.
There is no reason to assume that the size of the role set after these roles have been eliminated will be small, given that multiple roles may be assigned to the same permission in $P_{ub} \setminus P_{\ell b}$.

Although the results obtained in this paper are mainly negative, and those that are positive require strong restrictions on the \UAQ\ instances, we believe that the work provides useful insights into the difficulty of solving \UAQ.
In particular, we believe these results supplement the recent work of Armando {et al.}~\cite{ArmandoGT20} and may provide useful input into evaluating \UAQ\ solvers and producing new benchmarks for \UAQ.

It is well-known that practitioners prefer to use general-purpose solvers for solving practical problems rather than specialised algorithms.
It has been shown in the literature, see e.g. \cite{ArmandoGT20} for UAQ and \cite{KarapetyanPGG19} for the Workflow Satisfiability Problem, that appropriately chosen general-purpose solvers perform reasonably well on moderate-size instances of tractable problems.
Thus, it would be interesting to see whether the state-of-the-art PMaxSAT solver used in \cite{ArmandoGT20} performs well on instances of the UAQ problem of Section 5.

\paragraph{Acknowledgement} {We are very thankful to the referees and Eduard Eiben for providing very helpful suggestions, which improved the presentation.}
Research in this paper was supported by Leverhulme Trust grant RPG-2018-161.



\end{document}